\documentclass[12pt]{article}
\newtheorem{theorem}{Theorem}[section]
\newtheorem{corollary}[theorem]{Corollary}

\newtheorem{proposition}[theorem]{Proposition}
\newtheorem{example}{Example}[section]

\usepackage[toc,page]{appendix}
\newenvironment{proof}{\noindent{\bf Proof:}}{\hfill\fbox{}\vspace*{1mm}}
\usepackage[usenames]{color}
\usepackage{amssymb}
\usepackage{setspace}
\usepackage{graphicx}
\usepackage{subfigure}
\usepackage{float}
\usepackage{booktabs}

\usepackage{titlesec}
\titleformat{\section}[hang]{\Large\bfseries\filcenter}{\thesection.}{1em}{}

\RequirePackage[normalem]{ulem} 
\providecommand{\DIFdeltex}[1]{{\protect\color{red}\sout{#1}}}                      
\usepackage{xcolor}
\usepackage{changebar}

\newif\ifdiff
\difftrue
\ifdiff
  \newcommand{\del}[1]{\DIFdeltex{#1}}

\else
  \newcommand{\del}[1]{}

\fi

\makeatletter
\@addtoreset{equation}{section}
\makeatother

\usepackage[leqno]{amsmath}

\begin{document}
\title{\bf Trading Strategy with Stochastic Volatility in a Limit Order Book
Market}
\date{}
\author{
Wai-Ki Ching
\thanks{Corresponding author. 
Advanced Modeling and Applied Computing Laboratory,
Department of Mathematics, 
The University of Hong Kong, Pokfulam Road, Hong Kong. 
E-mail: wching@hku.hk. }
\and Jia-Wen Gu
\thanks{Department of Mathematical Science,
University of Copenhagen, 
Denmark. 
E-mail: jwgu.hku@gmail.com.}
\and Tak-Kuen Siu
\thanks{ Department of Applied Finance and Actuarial Studies,
Faculty of Business and Economics,
Macquarie University, Sydney, NSW 2109, Australia.
Email: ktksiu2005@gmail.com}
\and Qing-Qing Yang
\thanks{ Advanced Modeling and Applied Computing Laboratory,
Department of Mathematics, 
The University of Hong Kong, Pokfulam Road, Hong Kong. 
E-mail: kerryyang920910@gmail.com.}
}
\maketitle
\begin{spacing}{1.5}
\begin{abstract}
In this paper, we employ the {\it Heston stochastic volatility model}
\cite{Heston} to describe the stock's volatility and apply the model to derive and analyze the optimal trading strategies for dealers in a security market.  
We also extend our study to option market making for options written on stocks in the presence of  stochastic volatility. 
Mathematically, the problem is formulated as a stochastic optimal control problem and the controlled state process is the dealer's mark-to-market wealth. 
Dealers in the security market can optimally determine their ask and bid quotes on the underlying stocks or options continuously over time.
Their objective is to maximize an expected profit
from transactions with a penalty proportional to the variance of cumulative inventory cost.

\end{abstract}

\section{Introduction}

The optimal trading strategy of dealers in a Limit Order Book (LOB) market
has been widely studied in early 1990s, see \cite{Gould} 
for a detailed survey.
Ho and Stoll (1981) \cite{Ho1} provided one of the early studies on the behavior of a monopolistic dealer
in a  single stock situation.
Avellaneda and Stoikov (2008) \cite{Avellaneda08}
proposed a quantitative model for LOB by making use of its statistical
properties together with the utility framework of Ho and Stoll.
Gu{\'e}ant et al. (2012) \cite{GLF} provided simple and easy-to-compute
expressions for optimal quotes when the trader is willing to liquidate a portfolio.
Regarding the option market making, recent research includes, for example, \cite{GP} and \cite{OM}.
Due to the tractable, theoretical and empirical appeal, it seems that  most of the studies, 
see for example, \cite{OM,Bollen97}, may perhaps be based on the assumption that the volatility of the
underlying security is constant over time or is independent of the changes 
in the price level of the underlying security.
However, the empirical characteristics, e.g., 
leverage effect, time scale variance, volatility smile, mean-reverting and volatility clustering, cast doubts on the constancy of volatility in the context of market micro structure. There are some researches, see for instance \cite{SM} for a detailed survey, studying the modeling of volatility, for example, \cite{Bouchaud,Gould,Hendershott,Wyart,Zumbach}, among which three categories of non-constant volatility models are mainly discussed, which include 
\begin{enumerate}
\item Time-dependent deterministic volatility $\sigma(t)$,
\item Local volatility: volatility dependent on the stock price $\sigma(S_t)$,
\item Stochastic volatility: volatility driven by an additional random process $\sigma(w)$.
\end{enumerate}
In this paper we adopt Heston's mean-reverting  stochastic volatility model   corresponding to an arithmetic Brownian motion to set up our model 
$$
\left\{
\begin{array}{lll}
dS_t=\sqrt{\nu_t}dW_t\\
d\nu_t=\theta(\alpha-\nu_t)dt+\xi\sqrt{\nu_t}dB_t\\
\end{array}
\right.
$$
where $\it{W_t}$ and $\it{B_t}$ are correlated standard Brownian motions.  
Under this setting, we mainly study three different aspects of 
optimal trading in a Limit Order Book (LOB) market.
The quoting strategy for dealers in a LOB 
with stochastic volatility is first considered. 
We apply a combined approach of an asymptotic expansion and a linear approximation to reduce the resulting Hamilton-Jacobi-Bellman (HJB) equation to a series of Partial Differential Equations (P.D.E.s), which can be solved by using the Feynman-Kac formula.
Differences between the exact and the approximate value function as well as quotes are examined and discussed. 
Second, we extend the model to more general situations by taking the market impact into consideration.  
Three different types of market impact models are analyzed to shed light on the relationship between a trading strategy and the market impact. 
Third, we study an option market making strategy with stochastic volatility. Heston's model stands out from other stochastic volatility models here because there exist an analytical solution for European options that takes the correlation  between stock price and volatility into consideration \cite{Heston}. 
In this setting the market is incomplete due to the uncertainty from the source of volatility.
After taking  into account the market price of risk arising from  stochastic volatility, the optimal control problem is turned into a problem of 
solving an HJB equation, and the same method can then be employed to obtain an approximate solution.
In the case of option market making, different from the work in 
Stoikov and Saglam  \cite{OM}, an arbitrage-free price in the stochastic volatility model is used to set the option mid-price. 
Then the optimal bid and ask prices of the option are determined based on the option mid quote. 
We note that the market in the stochastic volatility model is incomplete
and there is more than one arbitrage-free price of the option, and hence, the option mid quotes. 
In other words, the option mid quote depends on the market price of risk, so do the optimal ask and bid quotes determined by the option mid quotes.

The paper is organized as follows. 
In Section 2, we introduce a fundamental model with stochastic volatility, under which we study optimal trading strategies in a setting of stock market making. 
The model is then generalized in Section 3 by incorporating the market impact factor, which can also be regarded as adverse selection. 
Three different types of models are analyzed here to explore the relationship between optimal trading strategies and market impact.  
In Section 4, we focus on the optimal trading strategies for options in a financial market with stochastic volatility. 
Both the case of market making in a stock and an option written on it simultaneously and the case of market making in the option with Delta-hedging are studied in this section.
Finally concluding remarks are given in Section 5.

\section{Stock Market Making in a Limit Order Book}

Technological innovation has completely changed  the role of a dealer,  especially with the growth of electronic exchanges such as Nasdaq's Inet.  
 Orders are placed in an automatic and electronic order-driven platform and wait in the {\it Limit Order Book (LOB)} to be executed. 

\subsection{Model Setup}

In this section, we consider a model to study the impact of stochastic volatility on dealer's optimal trading strategy.  
We assume that the stock mid-price evolves over time according to an arithmetic Brownian motion with stochastic volatility. 
More specifically, the Heston mean-reverting stochastic volatility model is adopted here:
\begin{equation}\label{Eq1}
\left\{
\begin{array}{lll}
dS_t = \sqrt{\nu_t} d W_t\footnotemark[1] \\
d\nu_t=\theta (\alpha-\nu_t)dt+\xi \sqrt{\nu_t}dB_t.
\end{array}
\right.
\end{equation}
\footnotetext[1]{When long-term strategies are considered, it is important to consider geometric rather than arithmetic Brownian motion. 
However, we focus on the trading strategies in a LOB, where trades are  high frequency and short-term. Therefore, the total fractional price changes are small, and the difference between arithmetic and geometric Brownian motions is negligible.}Here,  $d\nu_t$ modeling process originates from 
the CIR interest rate process \cite{CIR}.  
In the stochastic differential equation,  $\theta$, $\alpha$, and $\xi$ are positive constants.  
And $\{B_t\}$ and $\{W_t\}$ are two standard Brownian motions with constant correlation coefficient $\rho$ so that 
$$
W_t=\rho B_t+\sqrt{1-\rho^2} \tilde{B}_t
$$
where $\it{\{B_t\}}$ and $\{\tilde{B}_t\}$ are two independent Brownian motions. 
In \cite{Gould} and the references therein, it was pointed out that
a strong negative correlation between $S_t$ and
the realized mid-price volatility, i.e. leverage effect, has been observed in a wide range of markets, e.g., Paris Bourse \cite{Bouchaud}, FTSE 100 \cite{Zumbach}, 
and NYSE \cite{Wyart}, and our stochastic volatility model can well capture this characteristic. 

In Eq. (\ref{Eq1}), the drift term is zero, which means we have no information about the direction of future price movements. 
In fact, the drift is generally not significant over a short trading horizon.
Let $(S_t,\nu_t,W_t,B_t)=(s,\nu,0,0)$ be the initial state.
Some remarks of the stochastic volatility model are given below:
\begin{description}
\item[(i)] Although there does not exist a closed-form solution for 
Eq. (\ref{Eq1}), 
the model can ensure that the volatility is always nonnegative.  
Intuitively, when $\nu_t$ reaches zero, the coefficient of $dB_t$ vanishes
and the positive drift term  will drive
the volatility back to the positive territory.

\item[(ii)] Standard calculations give:
\begin{equation}\label{eq2}
E_t[\nu_u]=e^{-\theta(u-t)}\nu+\alpha\left(1-e^{-\theta(u-t)}\right).
\end{equation}
In particular, we have
$$
\lim_{u\to \infty}E_t[\nu_u]=\alpha
$$
i.e., $\alpha$ is the mean long-term volatility
and $\theta$ is the rate at which the volatility reverts toward its long-term mean. We also have
\begin{equation}\label{eq3}
\mbox{Var}[\nu_u|\mathcal{F}_t]=\frac{\xi^2}{\theta}\nu\left(e^{-\theta(u-t)}-e^{-2\theta(u-t)}\right)+\frac{\alpha\xi^2}{2\theta}\left(1-2e^{-\theta(u-t)}+e^{-2\theta(u-t)}\right).
\end{equation}
In particular, we have
$$
\lim_{u\to \infty}\mbox{Var}[\nu_u|\mathcal{F}_t]=\frac{\alpha\xi^2}{2\theta}.
$$
\end{description}
For more details about the standard calculations in these remarks, 
we refer readers to Appendix A1. 

\subsection{ State Feedback Control Problem}

In this section, we will use the above  setting to analyze the optimal trading strategies for dealers in the stock market.

\subsubsection{States and Controls}
Consider an active dealer in a LOB market, quoting a bid price $p_t^b$
and an ask price $p_t^a$ at no cost at time $t$, besides the prescribed minimum.
The dealer is committed to, respectively,
buy and sell one share of stock at these prices.
The wealth in cash, $X_t$, jumps whenever there is a buy
or sell order,
\begin{equation}\label{eq4}
d X_t =p_t^a d N_t^a - p_t^b d N_t^b
\end{equation}
where $N_t^b$ and  $N_t^a$ represent the amount of stocks bought and  sold by the dealer by the time $t$.
They are assumed to be independent {\it Poisson Processes}
with rates $\lambda_t^b$ and $\lambda_t^a$ respectively.
The number of stocks held is then given by $q_t = q_0+N_t^b-N_t^a$.
In addition, the arrival rates of buy and sell orders
that will reach the dealer depend on
the distances  from the current market price
$\delta_t^a=p_t^a-s$ and $\delta_t^b=s-p_t^b$, which  can be interpreted as the premiums for  the dealer to  sell and buy an unit of share in the LOB market.
Avellanda and Stoikov (2008) \cite{Avellaneda08}
aggregated all the statistical information of
LOBs and derived the trading intensity, 
which takes the following parametric form:
$$
\lambda ^b(\delta)=\lambda ^a(\delta)=A\exp(-k\delta).
$$

The mark-to-market wealth, $X_t+q_tS_t$, then follows  
$$
\begin{array}{lll}
d(X_t+q_tS_t)&=&\underbrace{\delta_t^a dN_t^a+ \delta_t^b dN_t^b}  \quad+\quad\underbrace{q_tdS_t}.\\
&& \quad(revenues)\;\quad(inventory\; value)\\
\end{array}
$$
Note that, $E[q_tdS_t]=0$  and
$$
E[d(X_t+q_tS_t)]=E[\delta_t^adN_t^a+\delta_t^bdN_t^b]
$$ 
i.e., the expected revenues from transactions equals to the expected excess returns with respected to the mark-to-market wealth 
(for more details,  please refer to Appendix A2).
Denote 
$$
dZ_t=\delta_t^adN_t^a+\delta_t^bdN_t^b \quad {\rm and} \quad 
dI_t=q_tdS_t
$$ 
 with $\{Z_t\}$, $\{I_t\}$ representing, respectively, the revenues from transactions and the inventory value. 

\subsubsection{The Objective}

Suppose the dealer in the LOB market is to liquidate $q$ shares of orders before time $T$ (a short time). 
If the $q$ orders are not completely executed at time $T$, then he has to sell the non-executed orders at the market price with certain clearing fee, $\beta/share$. 

We assume that  the  dealer is to maximize 
the expected mark-to-market wealth at time $T$,  with a penalty term arising from the  inventory value uncertainty.   At any time $t$,  we aim to find an optimal strategy by solving the following optimization problem:
$$
\displaystyle \max_{(\delta_u ^a,\delta_u ^b)_{u\in[t,T]}}\left\{E_t[Z_T-\beta q_T]-\frac{\gamma}{2} \mbox{Var}[I_T|\mathcal{F}_t]\right\},
$$
which is actually a stochastic state feedback control problem, and the martingale property of $\{I_t\}$ provides us a way to  further simplify this optimization problem. Thus, it is equivalent for us to find the optimal strategy for
$$
Z_t-\beta q_t+\displaystyle \max_{(\delta_u^a,\delta_u^b)_{u\in[t,T]}}E_t\left[\int_t^T(\delta_u^a+\beta)dN_u^a+(\delta_u^b-\beta)dN_u^b-\frac{\gamma}{2}\int_t^Tq_u^2\nu_udu\right].
$$
One key quantity in the model is 
\begin{equation}\label{E1}
V(q_t,\nu_t,t)=\displaystyle \max_{(\delta_u^a,\delta_u^b)_{u\in[t,T]}}E_t\left[\int_t^T(\delta_u^a+\beta)dN_u^a+(\delta_u^b-\beta)dN_u^b-\frac{\gamma}{2}\int_t^Tq_u^2\nu_udu\right].
\end{equation}
We denote it as our value function. Even though the actual value function is $Z_t-\beta q_t+V(q_t,\nu_t,t)$, given all the information up to  the time $t$, $Z_t-\beta q_t$ doesn't provide any useful information about the future states, so we just omit this term here. 

The other key quantity in the model is 
\begin{equation}\label{E2}
(\delta_u^{*,a},\delta_u^{*,b})_{u\in[t,T]}=arg \displaystyle  \displaystyle \max_{(\delta_u^a,\delta_u^b)_{u\in[t,T]}}E_t\left[\int_t^T(\delta_u^a+\beta)dN_u^a+(\delta_u^b-\beta)dN_u^b-\frac{\gamma}{2}\int_t^Tq_u^2\nu_udu\right]
\end{equation}
which is an optimal feedback  control process turning out  to be time and state dependent. 
\begin{proposition}
Suppose the value function (\ref{E1}) is sufficiently smooth, (i.e., $V \in {\cal C}^{1, 2}$). Then, the value function  satisfies the following HJB equation:
\begin{equation}\label{star}
\begin{array}{lll}
\displaystyle V_t+\theta(\alpha-\nu)V_{\nu}+ \frac{1}{2}\xi^2\nu V_{\nu\nu}-\frac{\gamma}{2}q^2\nu+\max_{\delta_t^a}\lambda^a(\delta_t^a)[\delta_t^a+\beta+V(q-1,\nu,t)-V(q,\nu,t)]\\
\displaystyle +\max_{\delta_t^b}\lambda^b(\delta_t^b)[\delta_t^b-\beta+V(q+1,\nu,t)-V(q,\nu,t)]=0\\
\end{array}
\end{equation}
with the boundary condition $V(q,\nu,T)=0$.
\end{proposition}
\begin{proof}
See Appendix B1.
\end{proof}

\begin{corollary}
The optimal controls at any time $t$ are given by 
\begin{equation}\label{control}
\begin{array}{lll}
(\delta_t^{*,a},\delta_t^{*,b})(q_t,\nu_t,t)&=&\displaystyle\Big(-\frac{\lambda_t^a}{\partial \lambda_t^a/\partial \delta_t^a}-\beta+V(q_t,\nu_t,t)-V(q_t-1,\nu_t,t),\\
&&\quad\displaystyle -\frac{\lambda_t^b}{\partial \lambda_t^b/\partial \delta_t^b}+\beta+V(q_t,\nu_t,t)-V(q_t+1,\nu_t,t)\Big)\\
\end{array}
\end{equation}
where  the value function, $V(q,\nu,t)$,  satisfies the following PDE
\begin{equation}\label{PDE}
\displaystyle V_t+\theta(\alpha-\nu)V_{\nu}+ \frac{1}{2}\xi^2\nu V_{\nu\nu}-\frac{\gamma}{2}q^2\nu-\frac{(\lambda_t^a)^2}{\partial \lambda_t^a/\partial \delta_t^a}-\frac{(\lambda_t^b)^2}{\partial \lambda_t^b/\partial \delta_t^b}=0
\end{equation}
with boundary condition $V(q,\nu,T)=0$.
\end{corollary}
\begin{proof}
Take the first-order optimality condition in Eq. (\ref{star}).
\end{proof}

\subsection{Optimal Quotes}

In this section, we focus on the computation of the optimal controls, which can be 
derived through an intuitive, two-step procedure. 
First, solve  Eq. (\ref{PDE}), then solve Eq. (\ref{control}). 
The main computational difficulty lies in solving Eq. (\ref{PDE}), 
since it not only contains continuous variables $t$ and $\nu$, but also  a discrete variable $q$.  However, due to our choice of ``{\it mean-variance}'' objective function, we are able to simplify the problem through an asymptotic expansion of $V(q,\nu,t)$ in the inventory variable $q$, which is  an approximative quadratic polynomial.  Before solving the problem, we first analyze an extreme case.
\begin{example}
For an inactive dealer who does not have any limit orders
in the market and simply holds an inventory of $q$ stocks
until the terminal time $T$, we have $dZ_t\equiv0$ and $ q_t\equiv q$. 
Then, by Eq. (\ref{E1}) 
\begin{equation}
\begin{array}{lll}
V(q_t,\nu_t,t)
&=& -\displaystyle \frac{\gamma}{2}E_t\left[\int_t^T q^2\nu_udu\right]\\
&=& -\displaystyle \frac{\gamma}{2}\int_t^T q_t^2E_t[\nu_u]du\\
&=&-\displaystyle\frac{\gamma q_t^2}{2\theta}(\nu_t-\alpha)\left[1-e^{-\theta(T-t)}\right]-\frac{\gamma q_t^2}{2}\alpha(T-t)\\
\end{array}
\end{equation}
which is independent of $\xi$.
We remark that when $\xi=\theta=0$, we have $dS_t=\sqrt{\nu}dW_t$ and
\begin{equation}
\begin{array}{lll}
V(q,\nu,t)&=&\displaystyle  -\lim_{\theta \to 0}\displaystyle\left(\frac{\gamma q_t^2}{2\theta}(\nu_t-\alpha)\left[1-e^{-\theta(T-t)}\right]+\frac{\gamma q_t^2}{2}\alpha(T-t)\right)\\
&=& -\displaystyle\frac{\gamma q_t^2}{2}\nu_t(T-t).
\end{array}
\end{equation}
\end{example}
It is the simplest trading strategy in LOBs. 
We shall adopt this strategy as a benchmark for making comparisons
with other strategies throughout this section.
\begin{theorem}
Assume the arrival rates of buy and sell orders that will reach the dealer take the exponential form:
$$
\lambda^a(\delta)=\lambda^b(\delta)=A\exp(-k\delta)
$$
For an active dealer, the derived  optimal ask and bid quotes $(\delta_t^{a,*},\delta_t^{b,*})$ can be approximated by $(\hat{\delta}_t^{a,*},\hat{\delta}_t^{b,*})$ under the approximate treatment  in 
 \cite{Avellaneda08}, which are given by
$$
\left\{
\begin{array}{lll}
\hat{\delta}_t^{a,*}&=&\frac{1}{k}-\beta -\left(\frac{\gamma}{2\theta}(\nu_t-\alpha)[1-e^{-\theta(T-t)}]+\frac{\gamma}{2}\alpha(T-t)\right)(2q_t-1),\\
\hat{\delta}_t^{b,*}&=&\frac{1}{k}+\beta +\left(\frac{\gamma}{2\theta}(\nu_t-\alpha)[1-e^{-\theta(T-t)}]+\frac{\gamma}{2}\alpha(T-t)\right)(2q_t+1).\\
\end{array}
\right.
$$
Moreover, 
$|\delta_t^{a,*}-\hat{\delta}_t^{a,*}| \ll 1$ and
$|\delta_t^{b,*}-\hat{\delta}_t^{b,*}| \ll 1$.
The approximated value function is given by 
$$
\tilde{V}(q_t,\nu_t,t)=-\frac{\gamma q_t^2}{2}\nu_t(T-t)
$$
which is equal to the value function of inactive dealer and the exact value function of the active dealer satisfies
$$
\tilde{V}(q_t,\nu_t,t)\le V(q_t,\nu_t,t)\le \tilde{V}(q_t,\nu_t,t)+c(T-t)
$$
where $c$ is a  positive constant.
\end{theorem}
\begin{proof}
See Appendix B2.
\end{proof}

\begin{example}
Let us take a risk-neutral dealer as an example. 
In this situation, $\gamma=0$ and the optimization problem can be written as follows:
$$
\begin{array}{lll}
V(q_t,\nu_t,t)&=&\displaystyle\max_{(\delta_u^a,\delta_u^b)_{u\in[t,T]}}E_t
\left[ \int_t^T(\delta_u^a+\beta)dN_u^a+(\delta_u^b-\beta)dN_u^b\right]\\
&=&\displaystyle\max_{(\delta_u^a,\delta_u^b)_{u\in[t,T]}}\int_t^T\left[(\delta_u^a+\beta)\lambda^a(\delta_u^a)+(\delta_u^b-\beta)\lambda^b(\delta_u^b)\right]dt.
\end{array}
$$
This expression attains its maximum when the optimal distances satisfy the following first order conditions:
\begin{equation}\label{eq22}
\lambda^a(\delta_u^a)+(\delta_u^a+\beta)\frac{\partial \lambda^a(\delta_u^a)}{\partial \delta_u^a}=0
\quad {\rm and} \quad 
\lambda^b(\delta_u^b)+(\delta_u^b-\beta)\frac{\partial \lambda^b(\delta_u^b)}{\partial \delta_u^b}=0.
\end{equation}
Thus we have $\delta_u^{a,*}\equiv \frac{1}{k}-\beta$ and $\delta_u^{b,*}\equiv \frac{1}{k}+\beta$. 
We note that,  for the existence of the maximizer, we need
\begin{equation}\label {condition}
\left\{
\begin{array}{c}
\displaystyle (\delta_u^a+\beta)\frac{\partial ^2 \lambda^a(\delta_u^a)}{\partial (\delta_u^a)^2}+2\frac{\partial \lambda^a(\delta_u^a)}{\partial \delta_u^a}\le0\\
\\
\displaystyle (\delta_u^b-\beta)\frac{\partial ^2 \lambda^b(\delta_u^b)}{\partial (\delta_u^b)^2}+2\frac{\partial \lambda^b(\delta_u^b)}{\partial \delta_u^b}\le0.
\end{array}
\right.
\end{equation}
It is straightforward to verify that $\delta_u^{a,*}\equiv \frac{1}{k}-\beta$ and $\delta_u^{b,*} \equiv \frac{1}{k}+\beta$ 
also satisfy the conditions, Eq. (\ref{condition}), for the maximizer and the value function equals
$$
V(q_t,\nu_t,t)=\frac{Ae^{-1}}{k}(e^{k\beta}+e^{-k\beta})(T-t).
$$
At the same time, setting $c=\frac{Ae^{-1}}{k}(e^{k\beta}+e^{-k\beta})$, 
which is a finite positive constant, and by the approximation method, 
we obtain  
\begin{equation}\label{eq24}
\tilde{V}(q_t,\nu_t,t)=0 \le V(q_t,\nu_t,t)\le \tilde{V}(q_t,\nu_t,t)+c(T-t) 
\end{equation}
and 
$$
\hat{\delta}_t^{a,*}= \frac{1}{k}-\beta =\delta_t^{a,*} 
\quad {\rm and} \quad \hat{\delta}_t^{b,*}=\frac{1}{k} + \beta=\delta_t^{b,*}.
$$ 
\end{example}

We then set a bid-ask spread for the dealer, which is given by 
\begin{equation}
\begin{array}{lll}
\displaystyle \hat{\delta}_t^{a,*}+\hat{\delta}_t^{b,*} &= &\displaystyle \frac{2}{k}+\displaystyle \frac{\gamma}{\theta}(\nu_t-\alpha)\left[1-e^{-\theta(T-t)}\right]+\gamma \alpha(T-t)
\end{array}
\end{equation}
and the price adjustment variable, $m_t$,  is defined by 
\begin{equation}
\begin{array}{lll}
m_t&=&\hat{\delta}_t^{a,*}-\hat{\delta}_t^{b,*} =\displaystyle -2\beta-2\left(\frac{\gamma}{\theta}(\nu_t-\alpha)\left[1-e^{-\theta(T-t)}\right]+\gamma\alpha(T-t)\right)q_t.
\end{array}
\end{equation}
We now give some remarks on the approximations.  
\begin{description}
\item [(i)]  Dependence on $\nu$:
$$
\left\{
\begin{array}{c}
\frac{\partial \hat{\delta}_t^{a,*}}{\partial \nu}<0 \quad,\quad\frac{\partial \hat{\delta}_t^{b,*}}{\partial \nu}>0, \quad \mbox{if $q_t>0$}\\
\frac{\partial \hat{\delta}_t^{a,*}}{\partial \nu}>0 \quad,\quad\frac{\partial \hat{\delta}_t^{b,*}}{\partial \nu}>0, \quad \mbox{if $q_t=0$}\\
\frac{\partial \hat{\delta}_t^{a,*}}{\partial \nu}>0 \quad,\quad\frac{\partial \hat{\delta}_t^{b,*}}{\partial \nu}<0, \quad \mbox{if $q_t<0$}\\
\end{array}
\right.
$$
and
$$
\frac{\partial (\hat{\delta}_t^{a,*}+\hat{\delta}_t^{b,*})}{\partial \nu}>0.
$$
The rationale behind this is that an increase in the variance $\nu_t$
will lead to an increase in the inventory risk. 
Hence, to reduce this risk, dealers having a long position will try to lower their bid and ask prices so as to encourage selling and discourage purchasing.
Similarly, dealers with a short position will try to raise prices to
encourage purchasing and to discourage selling.
As a conclusion, due to the increase in price risk, the bid-ask spread,
which reflects the risk a market maker is facing, widens.

\item[(ii)] We note that
$$
\begin{array}{lll}
\hat{\delta}_t^{a,*}+\hat{\delta}_t^{b,*}&= & \displaystyle \frac{2}{k}+\displaystyle \frac{\gamma}{\theta}(\nu_t-\alpha)\left[1-e^{-\theta(T-t)}\right]+\gamma \alpha(T-t).
\end{array}
$$
If $\xi=0$ and $ \theta \to 0$, then
\begin{equation}
\hat{\delta}_t^{a,*}+\hat{\delta}_t^{b,*}=\frac{2}{k}+\gamma \nu_t(T-t).
\end{equation}
 This corresponds to the results in \cite{Avellaneda08} where the variance $\nu$ is a constant.

\item [(iii)] 
\begin{description}
\item[(a)] Regarding an inactive dealer in the security market with the ``frozen inventory'' trading strategy, no limit order, simply holding an inventory of $q$ shares until the terminal time $T$,  the followings hold:
\begin{description}
\item[(a1)] The expected terminal wealth
$$
E_t[X_T+q_T(S_T-\beta)]=x_t+q_t(s_t-\beta).
$$
\item[(a2)] The value function
$$
V(q_t,\nu_t,t)
=-\displaystyle\frac{\gamma q_t^2}{2\theta}(\nu_t-\alpha)\left[1-e^{-\theta(T-t)}\right]
-\frac{\gamma q_t^2}{2}\alpha(T-t).
$$
\end{description}
\item[(b)]For an active dealer using the optimal inventory strategy, 
the followings hold:
\begin{description}
\item[(b1)] The expected terminal wealth
$$
E_t[X_T+q_T(S_T-\beta)]=x_t+q_t(s_t-\beta)+E_t\left[\int_t^T (\hat{\delta}_u^{a,*}+\beta)dN_u^a +(\hat{\delta}_u^{b,*}-\beta) dN_u^b  \right].
$$
\item[(b2)] The value function
$$
V(q_t,\nu_t,t) \ge-\displaystyle\frac{\gamma q_t^2}{2\theta}(\nu_t-\alpha)\left[1-e^{-\theta(T-t)}\right]-\frac{\gamma q_t^2}{2}\alpha(T-t).
$$
\end{description}
\end{description}
Here the last term in the expected terminal wealth
can be seen as the cumulated returns from transactions
$$
\begin{array}{lll}
E_t\left[\int_t^T (\hat{\delta}_u^{a,*}+b)dN_u^a+(\hat{\delta}_u^{b,*}-b)dN_u^b \right]&=&
\int_t^T E_t\left[(\hat{\delta}_u^{a,*}+\beta)dN_u^a
+(\hat{\delta}_u^{b,*}-\beta)dN_u^b\right]\\
&=&\int_t^T \left((\hat{\delta}_u^{a,*}+\beta)\lambda^a(\hat{\delta}_u^{a,*})+(\hat{\delta}_u^{b,*}-\beta)\lambda^b(\hat{\delta}_u^{b,*})\right)du\\
&=&A\int_t^T \left((\hat{\delta}_u^{a,*}+\beta)e^{-k\hat{\delta}_u^{a,*}}
+(\hat{\delta}_u^{b,*}-\beta)e^{-k\hat{\delta}_u^{b,*}}\right)du.
\end{array}
$$
Since both variables $\hat{\delta}_u^{a,*}$ and $\hat{\delta}_u^{b,*}$, in the integrand,
are related to the variable $q_t$,
it is not easy to obtain a closed-form solution.
In the next section, we use Monte Carlo method to
verify that the last term is always positive.
Compared with ``frozen inventory strategy'',
our strategy can improve the final profit without falling below 
the dealer's original indifference curve
(in the situation of ``frozen inventory'' problem, $(\delta_t^a, \delta_t ^b)_{t \in[0,T]} $
can be seen as $(+\infty,+\infty)$),
which means that an active dealer always takes advantage over 
an inactive dealer.

\item[(iv)] The price adjustment,
\begin{equation}
\begin{array}{lll}
m_t&=&\displaystyle -2\beta-2\left(\frac{\gamma}{\theta}(\nu_t-\alpha)\left[1-e^{-\theta(T-t)}\right]+\gamma\alpha(T-t)\right)q_t
\end{array}
\end{equation}
approaches to $-2\beta-2\gamma \nu_t(T-t)q_t$, as $\theta \to 0$.
It depends on the inventory level and is an inventory response equation that specifies the price adjustments variable be negative (positive)
when the inventory is greater (less) than a certain amount of inventory.
Due to the liquidation (clearing) cost, 
a dealer with a large amount of inventory have an
urgent need to clear his holding during the trading period.
When $m_t<0$, both the bid price and ask price are ``low'' and the dealer has
an incentive to sell rather than to purchase, and as a result, it reduces
the dealer's inventory level.
When $m_t>0$ the dealer has an incentive to purchase rather than to sell, 
and as a result, it will raise his inventory level.
The degree of price response to an inventory change depends on the same factors determining the size of the bid-ask spread-time remained ($T-t$),
dealer's risk aversion (determined by $\gamma$) and volatility
(determined by $t,\nu$,$\theta$, and $\alpha$).
\end{description}

\subsection{Numerical Experiments}

In our numerical simulations, we adopt the following parameters 
which as the same as those in \cite{Avellaneda08}:
$s=100, T=1, \nu=4, dt=0.005, q=0, \gamma=0.1, \theta=0.02, \alpha=4, \rho=0.7, \xi=0.5, k=1.5, A=140$.
The simulations are obtained through the following procedure:
\begin{table}[!hbtp]
\centering
\begin{tabular}{ll}
\hline
Step 1:& Compute the agent's quotes $\delta^a$, $\delta^b$ and other state variables. \\
Step 2:
& With probability $\lambda^a(\delta^a)dt$, $dN^a=1, dX= s+\delta^a$;\\
& With probability $\lambda^b(\delta^b)dt$, $dN^b=1,dX= -s+\delta^b$;\\
& The mid-price is updated  by a random increment $\pm\sqrt{v}\sqrt{dt}$;\\
& The volatility is updated accordingly by a random increment: \\
&
$\theta(\alpha-v)dt+ \xi\sqrt{v}(\rho\sqrt{t}\pm\sqrt{1-\rho^2}\sqrt{t})$
\ or \ $\theta(\alpha-v)dt+ \xi\sqrt{v}(-\rho\sqrt{t}\pm\sqrt{1-\rho^2}\sqrt{t})$.\\
Step 3: & $t:=t+dt$, and return to Step $1$.\\
\hline
\end{tabular}
\end{table}
\begin{center}
[Figure $1$ here]
\end{center}
We first  use  Mentor Carlo Method to simulate the dynamics of the stock mid-price, which is show in  Figure $1$   in red curve, and then same  method is used to test the performance of the optimal trading strategy.  The curve in blue shows the dynamics of the ask price and the curve in green  shows the dynamics of the bid price.  As we can see from Figure 1, ask prices are always above the  stock mid-prices, bid prices are always below the stock-mid prices and they are believed to be mean-reverting.
Figure $2$ shows some detailed  information about the cumulated revenues from transactions and the ask-bid spread with respect to the optimal trading strategy.   Figure $2(a)$ shows that the optimal trading strategy can make a positive revenue for its users.  Ask-bid spreads are always used to describes  the risks one  faces in the security market, as we can see from Figure $2(b)$, the risks one faces roughly decrease with respect to the time. 
\begin{center}
[Figure $2$ here]
\end{center}

\subsubsection{Trading Curve and Risk Aversion }

The average number of shares at each point of time, 
say the trading curve,  
can be computed by using Monte-Carlo simulations.
Figure $3(a)$ and $3(b)$ depict the trading curves with initial 
inventory $q_0=6$ and $q_0=-6$, respectively, when 
the optimal trading strategy is adopted.
\begin{center}
[Figure $3$ here]
\end{center}

We notice that the average number of shares at the terminal time $T$ is not equal to $0$, which can be explained by a weak incentive of the trader to liquidate the trading position strictly before time $T$ due to the low clearing fee caused by the liquidation at the terminal time. 
There are some cases for which liquidation is not completed before time $T$.
\begin{center}
[Figure $4$ here]
\end{center}

Figure $4$  shows the effect of risk-aversion on the trading strategies.   
In Case (1), we have $\gamma=0.01$, 
which presents a trader who is  risk-neutral.  
The trader postpones liquidation and eventually in the position of short selling.  
In Case (3),  $\gamma= 1.00$, which describes  a risk-averse trader who wishes to sell quickly to reduce exposure to volatility risk.  
In Case (2), we have $\gamma=0.10$, which lies between the above two extremes. 
We can see from Figure $4$  that traders with $\gamma=0.10$  prefers to 
liquidate quickly to avoid risk arising from stock's volatility.

\subsubsection{Efficient Frontier}

The efficient frontier consists of all optimal trading strategies. 
Here the  ``optimal'' refers to the situation where no strategy has a smaller variance for 
the same or higher level of expected transaction profits, i.e., 
the optimal trading strategy solves the following constrained  
optimization  problem:
$$
\displaystyle \max_{(\delta_u^a,\delta_u^b)_{u\in [t,T]}: Var[I_T|\mathcal{F}_t]\le V_*}E_t[Z_T-\beta q_T]
$$
for some $V_*$. 
We can solve the constrained optimization problem by introducing a Lagrange multiplier $\lambda$, i.e., solving the unconstrained problem,
$$
\max_{(\delta_u^a,\delta_u^b)_{u\in[t,T]}}\{E_t[Z_T-\beta q_T]-\lambda \mbox{Var}[I_T|\mathcal{F}_t]\}.
$$
For each level of $\lambda$, 
corresponding to a certain risk aversion, 
there is an optimal quoting strategy, given by $(\delta_t^a,\delta_t^b)_{t\in[0,T]}$. 
By running 1000 simulations with initial inventory $q_0=6$, we obtain an efficient frontier (see Figure $5$). 
This frontier is increasing along an approximative smooth concave curve 
when the level of dealer's risk aversion decreases.  
It shows the tradeoff between the expected revenues from transactions and the cumulated variance of the inventory value.  
The point on the most right of Figure $5$ is obtained for a risk neutral dealer ($\gamma=0$), and we define this point as $(V_0,R_0)$.   
For any other point on the left, we have
$$
R-R_0\approx \frac{1}{2}(V-V_0)^2\frac{d^2 R}{d V^2}\Big|_{V=V_0}.
$$
A crucial insight is that for a risk neutral dealer, a first-order decrease in the expected revenues can incur a second-order increase in cumulated variance. 
\begin{center}
[Figure $5$ here]
\end{center}

\subsubsection{The Optimal Inventory Strategy and the Symmetric Strategy}

By running $1000$ simulations with initial inventory equal to zero,
we obtain a comparison  between our ``inventory'' strategy and the ``symmetric'' strategy,
which employs the average spread of our inventory strategy, but centered it at the mid-price, regardless of the inventory. 
Results are presented in Table 1.
\begin{center}
[Table $1$ here]
\end{center}

We can learn from the table that the symmetric strategy has a higher return and a larger standard derivation than those of the inventory strategy.
It is not difficult to understand that the symmetric strategy 
results in a slightly higher return than the inventory strategy since it is centered around the mid-price, and therefore receives a higher volume of orders than the inventory strategy. However, the inventory strategy obtains a 
Profit \& Loss ($P\&L$) profile with a much smaller variance,
which can be seen from the simulation results in Table 1. 
\begin{center}
[Figure $6$ here]
\end{center}

Figure $6$ depicts the distributions of the $P\&L$ from the two strategies. From Figure $6$, it seems that the distribution of the $P\&L$ of the symmetric strategy has a heavier tail than that of the inventory strategy.

\section{An Extension of the Model to the Case with Market Impact}

As mentioned in the work of Almgren \cite{RA1,RA2}, the price received on each trade is affected by the rates of buying and selling.  
An extension of the model by introducing market impacts is discussed in this
section.

We assume that the stock mid-price evolves according to the following dynamics:
\begin{equation}\label{ww}
dS_t=\sqrt{\nu_t}dW_t-\eta(t)(dN_t^b-dN_t^a)
\end{equation}
where $\eta(t)$ is a  function of $t$, representing the market impact, and  may be related to the states $S_t$ and $v_t$. 

The function $\eta(t)$ in Eq. (\ref{ww}) could be chosen to reflect any preferred model of market micro-structure, 
subject only to certain natural convexity conditions.  

\subsection{Constant Market Impact}

To study the market impact,
one simple way is to consider the following dynamics for the price \cite{RC}:
\begin{equation}\label{eq30}
dS_t=\sqrt{\nu_t}dW_t + \eta (dN_t^a- dN_t^b),
\end{equation}
where $\eta>0$ is a constant, describing the market steady situation.
In this model, the reference price decreases when a limit order on the bid side is filled, increases when a limit order on the ask side is filled and the amount of price increases or decreases is equal to the constant $\eta$.  
This is in line with the classical modeling of market impact for market orders.  
Therefore, the dealer's states follow the following process:
\begin{equation}\label{eq31}
\begin{array}{lll}
d(X_t+q_tS_t) &=&\underbrace{ \delta_t^adN_t^a+\delta_t^bdN_t^b}\quad +\quad \underbrace{q_tdS_t}.\\
&&\quad(revenues)\quad\quad(inventory\;value)\\
\end{array}
\end{equation}
Decompose this state process into two components:
\begin{description}
\item[(i)] The revenues from transactions
$$
dZ_t=\delta_t^adN_t^a+\delta_t^bdN_t^b.
$$
\item[(ii)] The inventory value
(in this section, we use the quadratic variation to describe the risk)
$$
dI_t=q_tdS_t.
$$
\end{description}
Then, we consider the following optimization problem:
$$
\displaystyle
\max_{(\delta_u ^a,\delta_u ^b)_{u\in[t,T]}} \left\{E_t[Z_T-\beta q_T]-\frac{\gamma}{2}E_t\left[\int_t^T\left(dI_u\right)^2\right]\right\}.
$$
Note that $\{N_t^a\}$ and $\{N_t^b\}$ are two independent {\it Poisson processes}, 
The table below gives the multiplication results from standard stochastic calculus.
\begin{table}[!hbtp]
\centering
\begin{tabular}{|c|c|c|c|c|}
\hline$\cdot$&$dt$&$dW_t$&$dN_t^a$&$dN_t^b$\\
\hline$dt$&0&$0$&$0$ &0\\
\hline$dW_t$&0&$dt$& 0&$0$\\
\hline$dN_t^a$&0&$0$& $dN_t^a$&0\\
\hline$dN_t^b$&0&0& 0&$dN_t^b$\\
\hline
\end{tabular}
\end{table}
\\
Consequently, from the above table,
$$
\begin{array}{lll}
(dI_t)^2
&=&q_t^2\nu_tdt+q_t^2\eta ^2\left(dN_t^a+dN_t^b\right).
\end{array}
$$
Our first model can be recovered by assuming $\eta =0$, in which 
$$
E_t\left[\int_t^T(dI_u)^2\right]=Var[I_T|\mathcal{F}_t].
$$
Thus, the optimization problem can be written as,
$$
\begin{array}{lll}
 \displaystyle Z_t-\beta q_t+ \max_{(\delta_u ^a,\delta_u ^b)_{u\in[t,T]}}
E_t \Big[\int_t^T[(\delta_u^a+\beta)dN_u^a+(\delta_u^b-\beta)dN_u^b]-\frac{\gamma}{2}\int_t^Tq_u^2\nu_udu\\
\quad\quad\quad\quad\quad\quad\quad\quad\quad\quad\quad\quad
\displaystyle -\frac{\gamma}{2}\int_t^Tq_u^2\eta ^2\left(dN_u^a+dN_u^b\right)\Big].
\end{array}
$$
One key quantity for the model is 
\begin{equation}\label{MI1}
\begin{array}{lll}
V(q_t,\nu_t,t)=\displaystyle \max_{(\delta_u ^a,\delta_u ^b)_{u\in[t,T]}}
E_t \Big[\int_t^T[(\delta_u^a+\beta)dN_u^a+(\delta_u^b-\beta)dN_u^b]-\frac{\gamma}{2}\int_t^Tq_u^2\nu_udu\\
\quad\quad\quad\quad\quad\quad\quad\quad\quad\quad\quad\quad
\displaystyle -\frac{\gamma}{2}\int_t^Tq_u^2\eta ^2\left(dN_u^a+dN_u^b\right)\Big].
\end{array}
\end{equation}
We denote it as our value function.

The other key quantity for the model is 
\begin{equation}\label{MM2}
\begin{array}{lll}
(\delta_u^{*,a},\delta_u^{*,b})_{u\in[t,T]}=arg \displaystyle \max_{(\delta_u ^a,\delta_u ^b)_{u\in[t,T]}}
E_t \Big[\int_t^T[(\delta_u^a+\beta)dN_u^a+(\delta_u^b-\beta)dN_u^b]-\frac{\gamma}{2}\int_t^Tq_u^2\nu_udu\\
\quad\quad\quad\quad\quad\quad\quad\quad\quad\quad\quad\quad
\displaystyle -\frac{\gamma}{2}\int_t^Tq_u^2\eta ^2\left(dN_u^a+dN_u^b\right)\Big]\\
\end{array}
\end{equation}
which is an optimal control process turning out to be time and state dependent. 

\begin{proposition}
The value function (\ref{MI1}) satisfies the following HJB equation:
\begin{equation}\label{star1}
\begin{array}{lll}
\displaystyle V_t+\theta(\alpha-\nu)V_{\nu}+ \frac{1}{2}\xi^2\nu V_{\nu\nu}-\frac{\gamma}{2}q^2\nu+\max_{\delta_t^a}\lambda^a(\delta_t^a)[\delta_t^a+\beta-\frac{\gamma\eta^2}{2}(q-1)^2+V(q-1,\nu,t)-V(q,\nu,t)]\\
\displaystyle +\max_{\delta_t^b}\lambda^b(\delta_t^b)[\delta_t^b-\beta-\frac{\gamma\eta^2}{2}(q+1)^2+V(q+1,\nu,t)-V(q,\nu,t)]=0\\
\end{array}
\end{equation}
with the boundary condition $V(q,\nu,T)=0$.
\end{proposition}
\begin{proof}
Similar to Proposition 1.
\end{proof}

\begin{corollary}
The optimal controls (\ref{MM2}) at any time $t$  are given by 
\begin{equation}\label{controlM}
\begin{array}{lll}
(\delta_t^{*,a},\delta_t^{*,b})(q_t,\nu_t,t)&=&\displaystyle\Big(-\frac{\lambda_t^a}{\partial \lambda_t^a/\partial \delta_t^a}-\beta+\frac{\gamma\eta^2}{2}(q_t-1)^2+V(q_t,\nu_t,t)-V(q_t-1,\nu_t,t),\\
&&\quad\displaystyle -\frac{\lambda_t^b}{\partial \lambda_t^b/\partial \delta_t^b}+\beta+\frac{\gamma\eta^2}{2}(q_t+1)^2+V(q_t,\nu_t,t)-V(q_t+1,\nu_t,t)\Big)\\
\end{array}
\end{equation}
where the value function,  $V(q,\nu,t)$, satisfies the following PDE
\begin{equation}\label{PDEM}
\displaystyle V_t+\theta(\alpha-\nu )V_{\nu}+ \frac{1}{2}\xi^2\nu V_{\nu\nu}-\frac{\gamma}{2}q_t^2\nu -\frac{(\lambda_t^a)^2}{\partial \lambda_t^a/\partial \delta_t^a}-\frac{(\lambda_t^b)^2}{\partial \lambda_t^b/\partial \delta_t^b}=0
\end{equation}
with boundary condition $V(q,\nu,T)=0$.
\end{corollary}
\begin{proof}
Take the first-order optimality condition in Eq. (\ref{star1}).
\end{proof}

\subsection{Optimal Quotes}

Similar to the first model, through an intuitive, two-step procedure and some approximative methods,  we can get the optimal quotes under market impact model. 
\begin{theorem}
Assume the arrival rates of buy and sell orders that will reach the dealer take the exponential form: $\lambda^a(\delta)=\lambda^b(\delta)=A\exp(-k\delta)$.
For an active dealer, the derived optimal ask and bid quotes $(\delta_t^{a,*},\delta_t^{b,*})$ can be approximated by $(\hat{\delta}_t^{a,*},\hat{\delta}_t^{b,*})$ under the approximate treatment in \cite{Avellaneda08}, which are given by
$$
\left\{
\begin{array}{lll}
\hat{\delta}_t^{a,*}&=&\frac{1}{k}-\beta+\frac{\gamma\eta^2}{2}(q_t-1)^2 -\left(\frac{\gamma}{2\theta}(\nu_t-\alpha)[1-e^{-\theta(T-t)}]+\frac{\gamma}{2}\alpha(T-t)\right)(2q_t-1),\\
\hat{\delta}_t^{b,*}&=&\frac{1}{k}+\beta +\frac{\gamma\eta^2}{2}(q_t+1)^2+\left(\frac{\gamma}{2\theta}(\nu_t-\alpha)[1-e^{-\theta(T-t)}]+\frac{\gamma}{2}\alpha(T-t)\right)(2q_t+1).\\
\end{array}
\right.
$$
\end{theorem}

\begin{proof}
See Appendix C.
\end{proof}

\subsection{Numerical Experiments}

In our numerical simulations, we adopt the following parameters, which have been used in Section 2.4:
$s=100, T=1, \nu=4, dt=0.005, q=0, \gamma=0.1, \theta=0.02, \alpha=4, \rho=0.7, \xi=0.5, \eta=0.09, k=1.5, A=140$.

The corresponding figures and data with respect to the previous model are presented below:
\begin{center}
[Figure 7 and 8 here]
\end{center}

Market impact has been taken into account with the results depicted in Figure 7 when deciding the optimal trading strategy, i.e., $(\delta_t^a,\delta_t^b)_{t\in[0,T]}$. 
We can see from this figure, when comparing with the previous model (Figure $1$), the price adjustment becomes more sensitive to the inventory risk after considering the market impact, e.g. at time $0.30$,  the bid price will cross the stock mid-price and at time $0.68$, the ask price will  across the stock mid-price, which means dealers in this security market prefer to reduce their exposure to the volatility risk at the cost of revenues from transactions. 
However, this factor does not lead to any significant change 
in the general trend of the cumulated revenues and the ask-bid spread, 
which can be seen from Figures $8(a)$ and  $8(b)$. 

As for the trend of trading curves, similar conclusions can be drawn except the liquidation speeds.  
More information about the trading strategy  is included in the following two figures.
\begin{center}
[Figure 9 and 10 here]
\end{center}

Running $1000$ simulations with initial inventory equal to zero to compare the performances of the ``inventory'' strategy and the ``symmetric'' strategy, the following results are obtained:
\begin{center}
[Table 2 here]
\end{center}
\begin{center}
[Figure 11 here]
\end{center}

From Table $2$,  the profit generated from the inventory strategy is lower than that generated from the symmetric strategy. However, the standard error of the former is lower than the latter. 
It means that there is less uncertainty in the profit generated by the inventory strategy than the symmetric one.

\subsection{Analysis of Two Special Cases}

In this section, we shall study two special cases.\\

\noindent
{\bf Coordinated Variation:}
Suppose $\nu_t$ and $\eta(t)$ vary perfectly inversely, e.g.,
$\nu_t\eta(t)\equiv c$ where $c$ is a constant.
Then, the dealer's optimization problem can be transformed into 
the following HJB equation:
\begin{equation}\label{H2}
\left\{
\begin{array}{lll}
V_t+\theta(\alpha-\nu)V_\nu+\frac{1}{2}\xi^2\nu V_{\nu\nu}-\frac{\gamma}{2}
\nu q^2+\displaystyle \max_{(\delta_t^a,\delta_t^b) }\Big\{\lambda_t^a[
\delta_t^a+b-\frac{\gamma c^2}{2\nu^2}(q-1)^2+V(q-1,\nu,t)\\
\displaystyle -V(q,\nu, t)]+\lambda_t^b[\delta_t^b-b-\frac{\gamma c^2}{2 \nu^2}(q+1)^2+V(q+1,\nu,t)-V(q,\nu,t)]\Big\}=0,\\
V(q,\nu,T)=0.
\end{array}
\right.
\end{equation}
Similar argument can also be used to analyze dealer's value function and optimal quotes, so we omit the details here.\\

\noindent
{\bf Two-variable Model:} We suppose that 
$\eta(t)=\bar{\eta}e^{\zeta(t)}$
and $\zeta(t)$ evolves over time according to the following process
$$
d\zeta(t)=a(t)dt+b(t)dB_L(t)
$$
where $a(t)$ and $b(t)$ are coefficients whose values may depend on $\eta$ and $\nu$. 
Here ${B_L(t)}$ and ${B(t)}$ are correlated standard Brownian motions, with a constant coefficient of correlation $\rho_0$ $(0<\rho_0<1)$. 
By assuming that the function $V(q,\nu,\eta, t)$ is sufficiently smooth, (i.e., $V\in {\cal C}^{1, 2,2}(t,\nu,\eta)$), using the same procedure as above yields the HJB equation 
for $V(q,\nu,\eta,t)$:
\begin{equation}\label{H1}
\left\{
\begin{array}{lll}
\displaystyle V_t+\theta(\alpha-\nu)V_\nu+\eta\left(a(t)+\frac{b(t)^2}{2}\right)V_\eta+
\frac{1}{2}\xi^2\nu V_{\nu\nu}+\xi\eta b(t)\sqrt{\nu}\rho_0 V_{\nu\eta}+\frac{1}{2}\eta^2b(t)^2V_{\eta\eta}-\frac{\gamma}{2}
\nu q^2\\
+\displaystyle \max_{(\delta_t^a,\delta_t^b) }\Big\{\lambda_t^a[
\delta_t^a+b-\frac{\gamma \eta^2}{2 }(q-1)^2+V(q-1,\nu,\eta,t)
-V(q,\nu,\eta,t)]\\
\displaystyle \quad \quad \quad\quad   +\lambda_t^b[\delta_t^b-b-\frac{\gamma \eta^2}{2 }(q+1)^2+V(q+1,\nu,\eta,t)-V(q,\nu,\eta,t)]\Big\}=0,\\
V(q,\nu,\eta,T)=0.\\
\end{array}
\right.
\end{equation}
By the standard It$\hat{\rm o}$'s product rule,
$$
\begin{array}{lll}
d\left(\nu_t\eta(t)\right)&=&\nu_td\eta(t)+\eta(t)d\nu_t+d\nu_td\eta(t)\\
&=&\eta\left(\theta(\alpha-\nu)+\nu\left(a(t)+\frac{b(t)^2}{2}\right)
\displaystyle+\xi b(t)\sqrt{\nu}\rho_0\right)dt+\eta\xi \sqrt{\nu}dB(t)+\nu\eta b(t)dB_L(t).
\end{array}
$$
Thus, the coordinated case, i.e. $d\left(\nu_t\eta(t)\right)=0$
can be recovered by making the following assumptions:
\begin{description}
\item[(i)] The Brownian motions $B_L(t)$ and $B(t)$ have a perfect positive correlation $\rho_0=1$.
\item[(ii)]
$$
\left\{
\begin{array}{lll}
\displaystyle a(t)=\frac{\xi^2-2\theta(\alpha-\nu)}{2\nu}\\
\displaystyle b(t)=-\frac{\xi}{\sqrt{\nu}}.
\end{array}
\right.
$$
\end{description}
Then the HJB equation (\ref{H1}) reduces to the HJB equation (\ref{H2}). 

\section{Equity Option Market Making}

In this section, we consider the market making of options in a financial market with stochastic volatility.
Both the case of market making in a stock and an option written on the stock simultaneously and the case of market making in the option with Delta-hedging assumption are studied.
In these cases, the price of the option depends on a variable that is not traded. Consequently, the risk-neutral valuation alone does not directly lead to a unique price of the option. 
A price of the option may be specified by adopting a market price of risk.

\subsection{Model Setup}

Suppose the mid-price of the underlying stock is governed by the 
Heston's mean-reverting stochastic volatility model (\ref{Eq1}).
We can restrict our attention to the European call option prices as European
put option prices follow from the well-known put-call parity: 
$$
C-P=S-Ke^{-r(T-t)}.
$$
Assume that the interest rate, $r$, equals $0$.
The dealer makes markets in a European call option
with maturity $T$ and strike $K$. 
By It$\hat{\rm o}$ formula, the option 
mid-price follows:
$$
dC(s,\nu,t)=\Theta_tdt+\Delta_tdS_t+\frac{1}{2}\Gamma_t(dS_t)^2+
C_\nu d\nu_t+\frac{1}{2}C_{\nu\nu}(d\nu_t)^2+C_{s\nu}dS_td\nu_t
$$
where $\Theta_t, \Delta_t$ and $\Gamma_t$ are the standard Greeks.

\begin{proposition}
Under the risk-neutral valuation method, option's  prices under stochastic volatility model (\ref{Eq1}) satisfies the following PDE
\begin{equation}\label{C1}
C_t+\frac{1}{2}\nu C_{ss}+\xi\rho\nu C_{s\nu}+\frac{1}{2}\xi^2\nu C_{\nu\nu}+[\theta(\alpha-\nu)-\xi\sqrt{\nu}\sqrt{1-\rho^2}\eta^{\nu}]C_{\nu}=0
\end{equation}
with boundary condition $C(s,\nu,T)=(s-K)^{+}$. Where, $\eta^{\nu}$ is the price of volatility risk not related to stock returns.
\end{proposition}
\begin{proof}
See Appendix D1.
\end{proof}

\begin{proposition}
Under the standard  Arbitrage Pricing Theory (APT) argument, the call price under stochastic volatility model (\ref{Eq1}) satisfies the following P.D.E.:
\begin{equation}\label{C2}
C_t+\frac{1}{2}\nu C_{ss}+\xi\rho\nu C_{s\nu}+\frac{1}{2}\xi^2\nu C_{\nu\nu}+[\theta(\alpha-\nu )-\lambda^{\nu}(s,\nu,t)]C_{\nu}=0
\end{equation}
with boundary condition $C(s,\nu,T)=(s-K)^+$. Where, $\lambda^{\nu}(s,\nu,t)$ is the price of volatility risk respected to  $d\nu_t$.

\end{proposition}
\begin{proof}
See Appendix D2.
\end{proof}

If we let $\lambda^{\nu}(s,\nu,t)=\xi\sqrt{\nu}\sqrt{1-\rho^2}\eta^{\nu}$, then $\lambda^{\nu}$ will be not related to the underlying stock price $S_t$. 
From the {\it Second fundamental theorem of asset pricing} \cite{Springer}, the option under the stochastic volatility model has more than one arbitrage-free prices, 
since under this setting, the market is incomplete. 

In this section, we directly apply the APT argument to price the option, 
i.e., the option price can be derived through the following P.D.E.:
 \begin{equation}
 \left \{
 \begin{array}{lll}
 \frac{1}{2}\nu C_{ss}+\xi\rho \nu C_{s\nu}+\frac{1}{2}\xi^2 \nu C_{\nu\nu}+[\theta(\alpha-\nu )-\lambda^{\nu}(s ,\nu,t)]C_\nu
+C_t=0\\
C(s,\nu,T)=(s-K)^{+}.
\end{array}
\right.
\end{equation}

In the following, we mainly  discuss two cases: market making in stocks and options simultaneously and market making in options with Delta-hedging assumption. To simplify the expressions, we simply set the clearing fees equal zero.

\subsection{Market Making in Stocks and Options Simultaneously}

The approach adopted here is to model the market maker's trading strategies of options in the same way as the underlying stocks as described in the previous section. 
In other words, the dealer will now control the premiums charged around the stock mid-price, $\delta_t^{a,s}$ and $\delta_t^{b,s}$, as well as around the option mid-price, $\delta_t^{a,o}$ and $\delta_t^{b,o}$  at no cost except some prescribed minimum where
$$
p_t^{a,o}=C_t+\delta_t^{a,o} \quad {\rm and} \quad  p_t^{b,o}=C_t-\delta_t^{b,o}.
$$
We assume that the number of options bought and sold before time $t $ can also be modeled by two independent {\it Poisson processes},
denoted by $N_t^{b,o}$ and $N_t^{a,o}$, respectively, with intensities:
$$
\lambda_t^{a,o}=Ae^{-k\delta_t^{a,o}} \quad {\rm and} \quad
\lambda_t^{b,o}=Ae^{-k\delta_t^{b,o}}.
$$
The mark-to-market wealth $\mathcal{W}_t$ is then given by
$$
\mathcal{W}_t=\Pi_t+q_t^s S_t+ q_t^o C_t
$$
where $\Pi_t$ is the wealth in cash. It follows that
$$
\begin{array}{lll}
d\mathcal{W}_t&=&\underbrace{\delta_t^{a,s}dN_t^{a,s}+
\delta_t^{b,s}dN_t^{b,s}+\delta_t^{a,o}dN_t^{a,o}+
\delta_t^{b,o}dN_t^{b,o}}+\underbrace{q_t^sdS_t+q_t^odC_t.}\\
&&\quad\quad\quad\quad\quad\quad\quad(revenues)\quad\quad\quad\quad\quad\quad\quad(inventory \;value)\\
\end{array}
$$
We may  decompose this wealth process into two parts: the revenues obtained from transactions, which follows
$$
dZ_t=\delta_t^{a,s}dN_t^{a,s}+
\delta_t^{b,s}dN_t^{b,s}+\delta_t^{a,o}dN_t^{a,o}+
\delta_t^{b,o}dN_t^{b,o}
$$
and the inventory value 
(we use its quadratic variance up to the terminal time $T$ to describe the inventory risk), which follows
$$
dI_t=q_t^sdS_t+q_t^odC_t. 
$$

The dealer now  is to set his bid and ask prices throughout the trading horizon to msximize the following  objective function:
$$
\max_{(\delta_u ^{a,s},\delta_u ^{b,s},
\delta_u ^{a,o},\delta_u ^{b,o})_{u\in[t,T]}} \left\{E_t[Z_T]-\frac{\gamma}{2}E_t \left[\int_t^T \left(dI_u\right)^2  \right]\right\} 
$$
where
$$
\begin{array}{lll}
\left(dI_u\right)^2
&=&\left[(q_u^s)^2+2q_u^sq_u^o(\Delta_u+\rho\xi C_\nu)+(q_u^o)^2(\Delta_u^2+2\rho\xi\Delta_uC_\nu+\xi^2C_\nu^2)\right]\nu_udu.
\end{array}
$$
Thus, the optimization problem can be written as,
$$
\begin{array}{lll}
 \displaystyle Z_t+ \max_{(\delta_u ^{a,s},\delta_u ^{b,s},
\delta_u ^{a,o},\delta_u ^{b,o})_{u\in[t,T]}}
E_t \Big[\int_t^T[\delta_u^{a,s}dN_u^{a,s}+\delta_u^{b,s}dN_u^{b,s}+\delta_u^{a,o}dN_u^{a,o}\\
\quad\quad\quad\quad\quad\quad\displaystyle+\delta_u^{b,o}dN_u^{b,o}]
-\frac{\gamma}{2}\int_t^T\Big[(q_u^s)^2+2q_u^sq_u^o(\Delta_u+\rho\xi C_\nu)\\
\quad\quad\quad\quad\quad\quad+\displaystyle (q_u^o)^2(\Delta_u^2+2\rho\xi\Delta_uC_\nu+\xi^2C_\nu^2)\Big]\nu_udu)\Big].
\end{array}
$$

One key quantity for the model is 
\begin{equation}\label{O1}
\begin{array}{lll}
V(s_t,\nu_t,q_t^s,q_t^o,t)=\displaystyle \max_{(\delta_u ^{a,s},\delta_u ^{b,s},
\delta_u ^{a,o},\delta_u ^{b,o})_{u\in[t,T]}}
E_t \Big[\int_t^T[\delta_u^{a,s}dN_u^{a,s}+\delta_u^{b,s}dN_u^{b,s}+\delta_u^{a,o}dN_u^{a,o}\\
\quad\quad\quad\quad\quad\quad\displaystyle+\delta_u^{b,o}dN_u^{b,o}]
-\frac{\gamma}{2}\int_t^T\Big[(q_u^s)^2+2q_u^sq_u^o(\Delta_u+\rho\xi C_\nu)\\
\quad\quad\quad\quad\quad\quad+\displaystyle (q_u^o)^2(\Delta_u^2+2\rho\xi\Delta_uC_\nu+\xi^2C_\nu^2)\Big]\nu_udu)\Big].
\end{array}
\end{equation}
We denote it as our value function.

The other key quantity for the model is 
\begin{equation}\label{O2}
\begin{array}{lll}
(\delta_{u,*}^{a,s},\delta_{u,*}^{b,s},\delta_{u,*}^{a,o},\delta_{u,*}^{b,o})_{u\in[t,T]}=arg \displaystyle  \max_{(\delta_u ^{a,s},\delta_u ^{b,s},
\delta_u ^{a,o},\delta_u ^{b,o})_{u\in[t,T]}}
E_t \Big[\int_t^T[\delta_u^{a,s}dN_u^{a,s}+\delta_u^{b,s}dN_u^{b,s}+\delta_u^{a,o}dN_u^{a,o}\\
\quad\quad\quad\quad\quad\quad\displaystyle+\delta_u^{b,o}dN_u^{b,o}]
-\frac{\gamma}{2}\int_t^T\Big[(q_u^s)^2+2q_u^sq_u^o(\Delta_u+\rho\xi C_\nu)\\
\quad\quad\quad\quad\quad\quad+\displaystyle (q_u^o)^2(\Delta_u^2+2\rho\xi\Delta_uC_\nu+\xi^2C_\nu^2)\Big]\nu_udu)\Big]\\
\end{array}
\end{equation}
which is an  optimal control process turning out to be time and state dependent. 

\begin{proposition}
Suppose $V$ is sufficiently smooth, (i.e., $V \in {\cal C}^{1, 2,2}(t,s,\nu)$, the value function (\ref{O1}) satisfies the following HJB equation:
\begin{equation}\label{HJBO1}
\begin{array}{lll}
\displaystyle V_t+\theta(\alpha-\nu)V_{\nu}+\frac{1}{2}\nu V_{ss}+\frac{1}{2}\xi^2\nu V_{\nu\nu}+\rho\xi\nu V_{s\nu}-\frac{\gamma}{2}\nu \Big[(q_t^s)^2+2q_t^sq_t^o(\Delta_t+\rho\xi C_{\nu})\\
+\displaystyle (q_t^o)^2(\Delta_t^2+2\rho\xi\Delta_t C_{\nu}+\xi^2 C_{\nu}^2)\Big]\\
\displaystyle +\max_{\delta_t^{a,s}}\lambda_t^{a,s}[\delta_t^{a,s}+V(s,\nu,q_t^s-1,q_t^o,t)
-V(s,\nu,q_t^s,q_t^o,t)]\\
+\displaystyle\max_{\delta_t^{b,s}}\lambda_t^{b,s}[\delta_t^{b,s}+
V(s,\nu,q_t^s+1,q_t^o,t)
-V(s,\nu,q_t^s,q_t^o,t)] \\
+\displaystyle\max_{\delta_t^{a,o}}\lambda_t^{a,o}[\delta_t^{a,o}+
V(s,\nu,q_t^s,q_t^o-1,t)
-V(s,\nu,q_t^s,q_t^o,t)]\\
+\displaystyle\max_{\delta_t^{b,o}}\lambda_t^{b,o}[\delta_t^{b,o}+V(s,\nu,q_t^s,q_t^o+1,t)
-V(s,\nu,q_t^s,q_t^o,t)]=0\\
\end{array}\end{equation}
with the boundary condition $V(s,\nu,q^s,q^o,T)=0$.
\end{proposition}
\begin{proof}
Similar to Proposition 1.
\end{proof}

\begin{corollary}
The optimal controls (\ref{O2}) at any time $t$  are given by 
\begin{equation}\label{controlO}
\begin{array}{lll}
(\delta_{t,*}^{a,s},\delta_{t,*}^{b,s},\delta_{t,*}^{a,o},\delta_{t,*}^{b,o})(s_t,\nu_t,q_t^s,q_t^o,t)&=&\displaystyle\Big(-\frac{\lambda_t^{a,s}}{\partial \lambda_t^{a,s}/\partial \delta_t^{a,s}}+V(s_t,\nu_t,q_t^s-1,q_t^o,t)-V(s_t,\nu_t,q_t^s,q_t^o,t),\\
&&\quad\displaystyle -\frac{\lambda_t^{b,s}}{\partial \lambda_t^{b,s}/\partial \delta_t^{b,s}}+V(s_t,\nu_t,q_t^s+1,q_t^o,t)-V(s_t,\nu_t,q_t^s,q_t^o,t),\\
&&\quad\displaystyle -\frac{\lambda_t^{a,o}}{\partial \lambda_t^{a,o}/\partial \delta_t^{a,o}}+V(s_t,\nu_t,q_t^s,q_t^o-1,t)-V(s_t,\nu_t,q_t^s,q_t^o,t),\\
&&\quad\displaystyle -\frac{\lambda_t^{b,o}}{\partial \lambda_t^{b,o}/\partial \delta_t^{b,o}}+V(s_t,\nu_t,q_t^s,q_t^o+1,t)-V(s_t,\nu_t,q_t^s,q_t^o,t)\Big)\\
\end{array}
\end{equation}
where the value function,  $V(s,\nu,q^s,q^o,t)$,  satisfies the following PDE
\begin{equation}\label{PDEO}
\begin{array}{lll}
\displaystyle V_t+\theta(\alpha-\nu)V_{\nu}+\frac{1}{2}\nu V_{ss}+\frac{1}{2}\xi^2\nu V_{\nu\nu}+\rho\xi\nu V_{s\nu}-\frac{\gamma}{2}\nu \Big[(q_t^s)^2+2q_t^sq_t^o(\Delta_t+\rho\xi C_{\nu})\\
\;
+\displaystyle (q_t^o)^2(\Delta_t^2+2\rho\xi\Delta_t C_{\nu}+\xi^2 C_{\nu}^2)\Big]-\frac{(\lambda_t^{a,s})^2}{\partial \lambda_t^{a,s}/\partial \delta_t^{a,s}}-\frac{(\lambda_t^{b,s})^2}{\partial \lambda_t^{b,s}/\partial \delta_t^{b,s}}-\frac{(\lambda_t^{a,o})^2}{\partial \lambda_t^{a,o}/\partial \delta_t^{a,o}}-\frac{(\lambda_t^{b,o})^2}{\partial \lambda_t^{b,o}/\partial \delta_t^{b,o}}=0\\
\end{array}
\end{equation}
with boundary condition $V(s,\nu,q^s,q^o,T)=0$.
\end{corollary}
\begin{proof}
Directly take the first-order optimality condition in Eq. (\ref{HJBO1}).
\end{proof}

\subsection{Optimal Quotes}

Similar to the first model, through an intuitive, two-step procedure and some approximative methods,  we can get the optimal quotes under this setting. 

\begin{theorem} Assume the arrival rates of buy and sell orders that will reach the dealer take the exponential form: $\lambda^{a,s}(\delta)=\lambda^{b,s}(\delta)=\lambda^{a,o}(\delta)=\lambda^{b,o}(\delta)=A\exp(-k\delta)$.
Let
$$
H_1(s_t,\nu_t,t)=-\gamma E_t \left[ \int_t^T\nu_u\Big(\Delta_u+\rho\xi C_\nu(S_u,\nu_u,u)\Big)du\right]
$$
and
$$
H_2(s_t,\nu_t,t)=-\frac{\gamma}{2} E_t \left[ \int_t^T\nu_u\Big(\Delta_u^2
+2\rho\xi\Delta_u  C_\nu(S_u,\nu_u,u)
+\xi^2 C_\nu^2(S_u,\nu_u,u)\Big)du \right],
$$
then the (approximate) optimal policy derived under the approximate treatments in \cite{Avellaneda08} for the dealer is given by
\begin{equation}
\left\{
\begin{array}{lll}
\hat{\delta}_{t,*}^{a,s}&=&\frac{1}{k}-\left(\frac{\gamma}{2\theta}(\nu_t-\alpha)
\left[1-e^{-\theta(T-t)}\right]+\frac{\gamma}{2}\alpha(T-t)\right)(2q_t^s-1)+H_1(s_t,\nu_t,t)q_t^o\\
\hat{\delta}_{t,*}^{b,s}&=&\frac{1}{k}+\left(\frac{\gamma}{2\theta}(\nu_t-\alpha)
\left[1-e^{-\theta(T-t)}\right]+\frac{\gamma}{2}\alpha(T-t)\right)(2q_t^s+1)-H_1(s_t,\nu_t,t)q_t^o\\
\hat{\delta}_{t,*}^{a,o}&=&\frac{1}{k}+H_2(s_t,\nu_t,t)(2q_t^o-1)+H_1(s_t,\nu_t,t)q_t^s\\
\hat{\delta}_{t,*}^{b,o}&=&\frac{1}{k}-H_2(s_t,\nu_t,t)(2q_t^o+1)-H_1(s_t,\nu_t,t)q_t^s.
\end{array}
\right.
\end{equation}
Moreover, the approximate value function is given by
\begin{equation}
\begin{array}{lll}
 \tilde{V}( s_t,\nu_t,q_t^s,q_t^o,t)&=&\displaystyle -\left(\frac{\gamma}{2\theta}(\nu_t-\alpha)
\left[1-e^{-\theta(T-t)}\right]+\frac{\gamma}{2}\alpha(T-t)\right)(q_t^s)^2\\
&&\displaystyle +H_1(s_t,\nu_t,t)q_t^sq_t^o+H_2(s_t,\nu_t,t)(q_t^o)^2
\end{array}
\end{equation}
which is equal to the value function of inactive trader 
(details are omitted, similar to the case in stock market making) and the exact value function
$$
V(s_t,\nu_t,q_t^s,q_t^o,t)\ge  \tilde{V}( s_t,\nu_t,q_t^s, q_t^o,t)
$$
\end{theorem}

\begin{proof}
See Appendix D3.
\end{proof}

\subsection{Market Making in Options with Delta-hedging Assumption}

In this section, the dealer in a LOB is supposed to continuously control his wealth in stock $\pi_t$ and the bid-ask premiums $\delta_t^{b,o}$ and $\delta_t^{a,o}$ on the option.
The mark-to-market wealth $\mathcal{W}_t$ is given by
\begin{equation}
\mathcal{W}_t=\pi_t+\Pi_t+q_t^{o}C_t
\end{equation}
where $\Pi_t$ is the wealth in cash and $\Pi_t$ jumps every time when there
is a buy or sell order
\begin{equation}
d\Pi_t=p_t^{a,o}dN_t^{a,o}-p_t^{b,o}dN_t^{b,o}
\end{equation}
so does the inventory in options,
\begin{equation}
dq_t^o=dN_t^{b,o}-dN_t^{a,o}.
\end{equation}
Due to the continuously adjusted inventory in stock and the Delta-hedging assumption, $q_t^s=-q_t^o\Delta_t$.  
Therefore,
$$
dq_t^s=-\Delta_tdq_t^o-q_t^od\Delta_t-dq_t^o d\Delta_t=-\Delta_tdN_t^{b,o}+\Delta_tdN_t^{a,o}-q_t^od\Delta_t.
$$ 
The mark-to-market wealth follows:
$$
\begin{array}{lll}
d\mathcal{W}_t&=&\underbrace{\delta_t^{a,o}dN_t^{a,o}+\delta_t^{b,o}dN_t^{b,o}}
+\underbrace{q_t^sdS_t+q_t^odC_t.}\\
&&\quad\quad(revenues)\quad\quad\quad(inventory\; value)\\
\end{array}
$$
Decompose this wealth into two parts: the revenues obtained from transactions, which follows,
$$
dZ_t=\delta_t^{a,o}dN_t^{a,o}+\delta_t^{b,o}dN_t^{b,o}
$$
and the inventory value (we use its quadratic variance up to time $T$ to describe  the risk), which follows,
$$
dI_t=q_t^sdS_t+q_t^odC_t.
$$
Note that,
$$
\begin{array}{lll}
\left(dI_t\right)^2
&=&\left[(q_t^s)^2+2q_t^sq_t^o(\Delta_t+\rho\xi C_\nu)+(q_t^o)^2(\Delta_t^2+2\rho\xi\Delta_tC_\nu+\xi^2C_\nu^2)\right]\nu_tdt.
\end{array}
$$
Assume that the options are {\it Delta-hedged} at every point of time $t$, i.e., $q_t^s=-q_t^o\Delta_t$, we then obtain that
$$
\left(dI_t\right)^2=\nu_t\xi^2 C_\nu^2(q_t^o)^2.
$$
Assume the dealer aims to set his bid and ask prices
continuously over the time horizon to optimize the following  objective function:
$$
\displaystyle \max_{(\delta_u ^{a,o},\delta_u ^{b,o})_{u\in[t,T]}}
\left\{ E_t[Z_T]-\frac{\gamma}{2}E_t\left[\int_t^T \left(dI_u\right)^2\right] \right\}. 
$$ 
Then, the optimization problem can be written as,
$$
\begin{array}{lll}
\displaystyle Z_t + \max_{(\delta_u^{a,o},\delta_u ^{b,o})_{u\in[t,T]}} E_t\left[\int_t^T\delta_u^{a,o}dN_u^{a,o}+\delta_u^{b,o}dN_u^{b,o}-\frac{\gamma}{2}
\int_t^T\nu_u \xi^2 C_\nu^2(q_u^o)^2 du\right].
\end{array}
$$
One key quantity for the model is 
\begin{equation}\label{OO1}
V(s_t,\nu_t,q_t^o,t)=\displaystyle \max_{(\delta_u^{a,o},\delta_u ^{b,o})_{u\in[t,T]}} E_t\left[\int_t^T\delta_u^{a,o}dN_u^{a,o}+\delta_u^{b,o}dN_u^{b,o}-\frac{\gamma}{2}
\int_t^T\nu_u \xi^2 C_\nu^2(q_u^o)^2 du\right].
\end{equation}
We denote it as our value function.

The other key quantity for the model is 
\begin{equation}\label{OO2}
(\delta_{u,*}^{a,o},\delta_{u,*}^{b,o})_{u\in[t,T]}=arg \displaystyle \max_{(\delta_u^{a,o},\delta_u ^{b,o})_{u\in[t,T]}} E_t\left[\int_t^T\delta_u^{a,o}dN_u^{a,o}+\delta_u^{b,o}dN_u^{b,o}-\frac{\gamma}{2}
\int_t^T\nu_u \xi^2 C_\nu^2(q_u^o)^2 du\right]
\end{equation}
which is an optimal control process turning out to be time and state dependent.

\begin{proposition}
Suppose that $V$ is sufficiently smooth, the value function (\ref{OO1}) satisfies the following HJB equation:
\begin{equation}\label{HJBOO}
\begin{array}{lll}
\displaystyle V_t+\theta(\alpha-\nu)V_\nu+\frac{1}{2}\nu V_{ss}+\xi\rho \nu V_{s\nu}+\frac{1}{2}\xi^2 \nu V_{\nu\nu}
-\frac{\gamma}{2}\nu\xi^2 C_\nu^2(q_t^o)^2
\displaystyle +\max_{(\delta_t^{a,o},\delta_t^{b,o})\in\mathcal{A}}\Big\{\lambda_t^{a,o}[\delta_t^{a,o}+\\
V(s,\nu,q_t^o-1,t)-V(s,\nu,q_t^o,t)] +\lambda_t^{b,o}[\delta_t^{b,o}+V(s,\nu,q_t^o+1,t)-V(s,\nu,q_t^o,t)]\Big\}=0\\
\end{array}
\end{equation}
with the boundary condition $V(s,\nu,q^o,T)=0$.
\end{proposition}
\begin{proof}
Similar to Proposition 1.
\end{proof}

\begin{corollary}
The optimal controls (\ref{OO2}) at any time $t$ are given by 
\begin{equation}\label{controlOO}
\begin{array}{lll}
(\delta_{t,*}^{a,o},\delta_{t,*}^{b,o})(s_t,\nu_t,q_t^o,t)&=&\displaystyle\Big(-\frac{\lambda_t^{a,o}}{\partial \lambda_t^{a,o}/\partial \delta_t^{a,o}}+V(s_t,\nu_t,q_t^o,t)-V(s_t,\nu_t,q_t^o-1,t),\\
&&\quad\displaystyle -\frac{\lambda_t^{b,o}}{\partial \lambda_t^{b,o}/\partial \delta_t^{b,o}}+V(s_t,\nu_t,q_t^o,t)-V(s_t,\nu_t,q_t^o+1,t)\Big)\\
\end{array}
\end{equation}
where the value function, $V(s,\nu,q,t)$,  satisfies the following PDE
\begin{equation}\label{PDEOO}
\displaystyle V_t+\theta(\alpha-\nu)V_\nu+\frac{1}{2}\nu V_{ss}+\xi\rho \nu V_{s\nu}+\frac{1}{2}\xi^2 \nu V_{\nu\nu}
-\frac{\gamma}{2}\nu \xi^2 C_\nu^2(q_t^o)^2-\frac{(\lambda_t^{a,o})^2}{\partial \lambda_t^{a,o}/\partial \delta_t^{a,o}}-\frac{(\lambda_t^{b,o})^2}{\partial \lambda_t^{b,o}/\partial \delta_t^{b,o}}=0
\end{equation}
with boundary condition $V(s,\nu,q^o,T)=0$.
\end{corollary}

\begin{proof}
Directly take the first-order optimality condition in Eq. (\ref{HJBOO}).
\end{proof}

\subsection{Optimal Quotes}
Similar to the first model, through an intuitive, two-step procedure and some approximative methods,  we can get the optimal quotes under this setting. 
\begin{theorem}
 Assume the arrival rates of buy and sell orders that will reach the dealer take the exponential form: $\lambda^{a,o}(\delta)=\lambda^{b, o}(\delta)=A\exp(-k\delta)$ and the option's  position  is Delta-hedged at any  time  $t$.
Let
$$
M (s_t,\nu_t,t)=-\displaystyle \frac{\gamma}{2}\xi^2\left(E_t\left[\int_t^T\nu_uC_\nu^2(S_u,\nu_u,u)du\right]\right).
$$
The approximate optimal controls $(\hat{\delta}_{t,*}^{a,o},\hat{\delta}_{t,*}^{b,o})$ of the market maker derived under the approximate treatments in \cite{Avellaneda08} at time t are given by
\begin{equation}
\left\{
\begin{array}{lll}
\hat{\delta}_{t,*}^{a,o}&=&\frac{1}{k}+M_2(s_t,\nu_t,t)(2q_t^o-1)\\
\hat{\delta}_{t,*}^{b,o}&=&\frac{1}{k}-M_2(s_t,\nu_t,t)(2q_t^o+1).\\
\end{array}
\right.
\end{equation}
Moreover, the approximate value function is given by
\begin{equation}
\begin{array}{lll}
 \tilde{V}(s_t, \nu_t,q_t^o,t)&=&\displaystyle  - \frac{\gamma}{2}\xi^2\left(E_t\left[\int_t^T\nu_uC_\nu^2(S_u,\nu_u,u)du\right]\right)(q_t^o)^2
\end{array}
\end{equation}
which is identical to the value function of the inactive trader 
(details are omitted, similar to the case in stock market making) 
and the exact value function
$$
V(s_t,\nu_t,q_t^o,t)\ge \tilde{V}(s_t,\nu_t,q_t^o,t).
$$
\end{theorem}

\begin{proof}
The proof is similar to that of Theorem 3, 
we refer readers to Appendix D4 for more details.
\end{proof}

\section{Conclusions}

In this paper, we adopt a stochastic volatility model to describe the dynamics of the underlying stock's volatility  and derive mean-quadratic-variation optimal trading strategies for market making in both the stock and option markets.
In  our settings, whether it is stock market making and its extension after taking the market impact into account  or option market making, 
the dealer in the security market always has control over his bid and ask quotes and aims to maximize the  expected revenues while minimizing the  quadratic variation of the inventory value.  
A stochastic control approach is used to solve these optimization problems, and eventually these optimal control problems are transformed into one solving a series of Hamilton-Jacobi-Bellman (HJB) equations.  
Analytic approximations of the optimal bid and ask quotes are obtained, and Monte Carlo simulations are used to compare the optimal strategies to a ``zero-intelligence'' strategy.  
An important topic for future research may perhaps be developing accurate and efficient method to solve the resulting HJB equation. 
This is particularly important because the optimal trading strategy cannot be obtained without solving the resulting HJB equation. 
 Moreover, volatility tends to be correlated with high trading volume and company specific news(e.g. earning announcements), other important further research issues may  include taking into account  these empirical characteristics and  extending our model 
to more general cases, for example,the case of a trend in the price dynamics, the effect of news events on securities markets, the application of HMM in the  LOBs, and the case of a multiple-dealer competitive market \cite{Yang}.

\section{Appendix}

\subsection*{A1. Remarks on Stochastic Volatility Model}

The proofs presented here mainly involve the use of some standard techniques in stochastic calculus. Define $f(u,x)=e^{\theta(u-t)}x$ and use the It$\hat{\rm o}$-Doeblin formula
to compute
\begin{equation}\label{AP1}
\begin{array}{lll}
d\left(e^{\theta(u-t)}v_u\right)&=&df(u,v_u)\\
&=&f_u(u,v_u)du+f_v(u,v_u)dv_u+\frac{1}{2}f_{vv}(u,v_u)dv_udv_u\\
&=&\theta e^{\theta(u-t)}v_udu+\theta(\alpha-v_u) e^{\theta(u-t)}du+\xi e^{\theta(u-t)}\sqrt{v_u}dB_u\\
&=&\theta \alpha e^{\theta (u-t)}du +\xi e^{\theta(u-t)}\sqrt{v_u}dB_u.
\end{array}
\end{equation}
Integrating both sides of Eq. (\ref{AP1}) from $t$ to $u$,
we obtain
\begin{equation}
\begin{array}{lll}
e^{\theta(u-t)}v_u&=&v+\theta\alpha \int_t^u e^{\theta(s-t)}ds+\xi\int_t^u e^{\theta(s-t)}\sqrt{v_s}dB_s\\
&=&v+\alpha\left[e^{\theta(u-t)}-1\right]+\xi\int_t^u e^{\theta(s-t)}\sqrt{v_s}dB_s.
\end{array}
\end{equation}
Using the local martingale property of the stochastic integral  and some standard stopping arguments,
 
\begin{equation}
e^{\theta(u-t)}E_t[v_u]=v+\alpha\left[e^{\theta(u-t)}-1\right]
\end{equation}
or, equivalently,
\begin{equation}
E_t[v_u]=e^{-\theta(u-t)}v+\alpha\left[1-e^{-\theta(u-t)}\right].
\end{equation}
To compute the variance of $v_u$,
we set $X_u=e^{\theta(u-t)}v_u$, for which we have already computed
\begin{equation}
dX_u=\theta \alpha e^{\theta (u-t)}du +\xi e^{\theta(u-t)}\sqrt{v_u}dB_u=\theta \alpha e^{\theta (u-t)}du +\xi e^{\frac{\theta(u-t)}{2}}\sqrt{X_u}dB_u
\end{equation}
and
$$
E_t[X_u]=v+\alpha\left[e^{\theta(u-t)}-1\right].
$$
According to the It$\hat{\rm o}$-Doeblin formula with $f(x)=x^2$, we have
\begin{equation}\label{AP2}
\begin{array}{lll}
d(X_u^2)&=&2X_udX_u+dX_udX_u\\
&=&2\theta\alpha e^{\theta(u-t)}X_udu+2\xi e^{\frac{\theta(u-t)}{2}}X_u^{\frac{3}{2}}dB_u+\xi^2 e^{\theta(u-t)}X_udu.
\end{array}
\end{equation}
Integrating Eq. (\ref{AP2}) from $t$ to $u$, we get
\begin{equation}
X_u^2=v^2+(2\theta\alpha+\xi^2)\int_t^ue^{\theta(s-t)}X_sds
+2\xi\int_t^u e^{\frac{\theta(s-t)}{2}}X_s^{\frac{3}{2}}dB_s
\end{equation}
and taking expectation, using the local martingale property of a stochastic integral and some stopping arguments as well as the formula already
derived for ${\mbox E}_t[X_u]$, we obtain
\begin{equation}
\begin{array}{lll}
E_t[X_u^2]&=&v^2+(2\theta\alpha+\xi^2)\int_t^ue^{\theta(s-t)}E_t[X_s]ds\\
&=&v^2+(2\theta\alpha+\xi^2)\int_t^ue^{\theta(s-t)}\left(v+\alpha(e^{\theta(s-t)}-1)\right)ds\\
&=&v^2+\frac{2\theta\alpha+\xi^2}{\theta}(v-\alpha)\left(e^{\theta(u-t)}-1\right)+\frac{2\theta\alpha+\xi^2}{2\theta}
\alpha\left(e^{2\theta(u-t)}-1\right).
\end{array}
\end{equation}
Therefore, 
\begin{equation}
\begin{array}{lll}
E_t[v_u^2]&=&e^{-2\theta(u-t)}E_t[X_u^2]\\
&=&e^{-2\theta(u-t)}v^2+\frac{2\theta\alpha+\xi^2}{\theta}(v-\alpha)\left(e^{-\theta(u-t)}-e^{-2\theta(u-t)}\right)+\frac{2\theta\alpha+\xi^2}{2\theta} \alpha\left(1-e^{-2\theta(u-t)}\right).
\end{array}
\end{equation}
Finally, 
\begin{equation}
\begin{array}{lll}
Var[v_u|\mathcal{F}_t]&=&E_t[v_u^2]-\left(E_t[v_u]\right)^2\\
&=&\frac{\xi^2}{\theta}v\left(e^{-\theta(u-t)}-e^{-2\theta(u-t)}\right)+\frac{\alpha\xi^2}{2\theta}\left(1-2e^{-\theta(u-t)}+e^{-2\theta(u-t)}\right).
\end{array}
\end{equation}

\subsection*{A2. Martingale Property of $\{ I_u \}$.}

To prove  
$$
E[d(X_u+q_uS_u)]=E[\delta_u^adN_u^a+\delta_u^bdN_u^b]
$$
is equivalent to prove  $E[q_udS_u]=0$.
Let $dI_u=q_udS_u$, it is equivalent to prove $\{I_u\}$ is a martingale.
Since $dS_u=\sqrt{\nu_u}dW_u$, then $dI_u=q_u\sqrt{\nu_u}dW_u$.
Note that $\{q_u\sqrt{\nu_u}\}$ is adapted to the filtration $\{ {\cal F}_u \}$, and the arrival rates
$$
\lambda ^b(\delta_t^b)=A\exp(-k\delta_t^b)
\quad {\rm and} \quad \lambda ^a(\delta_t^a)=A\exp(-k\delta_t^a)
$$
are bounded from above.
We define the common bound to be $\Lambda$, i.e., $\lambda_t^a<\Lambda$ and $\lambda_t^b < \Lambda$ for all $t$. 
From now on, to simplify the notation, we write ${\mbox E}_t[\cdot]$ for the conditional expectation ${\mbox E}[\cdot|{\cal F}_t]$.
Thus, we have the following remarks.

\begin{description}
\item[(i)] The superposition of the processes $N_t^a$ and $N_t^b$, $N_t:=N_t^a+N_t^b$, is also a Poisson process with an intensity parameter ($\lambda_t^a+\lambda_t^b$), 
which is bounded from above by $2\Lambda$. We note that
    $$
    dq_t=dN_t^b-dN_t^a \le dN_t.
    $$
    Thus $ E[q_t^n]<\infty$ for any integer $n$;

\item[(ii)] Note that
$\{\sqrt{\nu_u}\}$ is a stochastic process driven by the standard Brownian motion $\{B_u\}$, with
$$
 E_t[\nu_u]=e^{-\theta(u-t)}\nu+\alpha\left(1-e^{-\theta(u-t)}\right)
$$
and
$$
\begin{array}{lll}
E_t[\nu_u^2]
&=&e^{-2\theta(u-t)}\nu^2+\frac{2\theta\alpha+\xi^2}{\theta}(\nu-\alpha)\left(e^{-\theta(u-t)}-e^{-2\theta(u-t)}\right)\\
&&+\frac{2\theta\alpha+\xi^2}{2\theta}
\alpha\left(1-e^{-2\theta(u-t)}\right).
\end{array}
$$
\end{description}
$\quad$\\
Therefore,
$$
E_t \left[\int_t^T q_u^2\nu_udu \right]
=\int_t^TE_t[q_u^2\nu_u]du\le\int_t^T(E_t[q_u^4])^{\frac{1}{2}}
 (E_t[\nu_u^2])^{\frac{1}{2}}du<\infty.
$$
The  equality follows from Fubini's theorem, and the inequality results from the H$\ddot{\rm o}$lder's inequality.
Then, we have proven that $\{I_u\}$ is a martingale.

\subsection*{B1. Proof of Proposition 1}
$$
\begin{array}{lll}
V(q_t,\nu_t,t)&=&\displaystyle \max_{(\delta_u^a,\delta_u^b)_{u\in[t,T]}}E_t\left[ \int_t^T (\delta_u^a+\beta)dN_u^a+(\delta_u^b-\beta)dN_u^b-\frac{\gamma}{2}\int_t^Tq_u^2\nu_udu\right]\\
&=&\displaystyle \max_{(\delta_u^a,\delta_u^b)_{u\in[t,T]}}E_t\left[\left( \int_t^{t+dt}+\int_{t+dt}^T\right) \left[(\delta_u^a+\beta)dN_u^a+(\delta_u^b-\beta)dN_u^b-\frac{\gamma}{2}q_u^2\nu_udu\right]\right]\\
&&\quad\quad\quad\quad\quad\quad\quad\quad\quad\quad\quad\quad\quad\quad\quad\quad\quad\quad\quad\quad\quad\quad\quad\quad\quad(iterated\;property)\\
&=&\displaystyle \max_{(\delta_u^a,\delta_u^b)_{u\in[t,T]}}E_tE_{t+dt}\left[\left( \int_t^{t+dt}+\int_{t+dt}^T\right) \left[(\delta_u^a+\beta)dN_u^a+(\delta_u^b-\beta)dN_u^b-\frac{\gamma}{2}q_u^2\nu_udu\right]\right]\\
\end{array}
$$
Note that $dq_t=dN_t^b-dN_t^b$.  In  $[t,t+dt)$, the probability of observing a jump in $N_t^a$ and a jump in $N_t^b$ are given by, $\lambda_t^a dt$  and  $\lambda_t^b dt$, respectively. 
$$
\begin{array}{lll}
P(dN_t^a=0,dN_t^b=1)&=&\lambda_t^bdt+o(dt)\\
P(dN_t^a=1,dN_t^b=0)&=&\lambda_t^adt+o(dt)\\
P(dN_t^a=0,dN_t^b=0)&=&1-\lambda_t^adt-\lambda_t^bdt+o(dt)\\
P(dN_t^a=1,dN_t^b=1)&=&o(dt)\\
\end{array}
$$
Then,
$$
\begin{array}{lll}
V(q_t,\nu_t,t)
&=&\displaystyle \max_{(\delta_u^a,\delta_u^b)_{u\in[t,T]}}E_tE_{t+dt}\left[\left( \int_t^{t+dt}+\int_{t+dt}^T\right) \left[(\delta_u^a+\beta)dN_u^a+(\delta_u^b-\beta)dN_u^b-\frac{\gamma}{2}q_u^2\nu_udu\right]\right]\\
&=&\displaystyle \max_{(\delta_t^a,\delta_t^b)}\Big\{\lambda_t^adt[\delta_t^a+\beta+E_t[V(q_t-1,\nu_t+d\nu_t,t+dt)]]\\
&&\quad\quad\quad+\lambda_t^bdt[\delta_t^b-\beta+E_t[V(q_t+1,\nu_t+d\nu_t,t+dt)]]\\
&&\quad\quad\quad+(1-\lambda_t^adt-\lambda_t^b)E_t[V(q_t,\nu_t+d\nu_t,t+dt)]\Big\}.
\end{array}
$$
By It$\hat{\rm o}$'s lemma, we have
$$
V(q,\nu_t+d\nu_t, t+dt)=V(q,\nu,t)+\Big(V_t+\theta(\alpha-\nu_t)V_{\nu}+\frac{1}{2}\xi^2\nu_tV_{\nu\nu}\Big)dt+\xi\sqrt{\nu_t}V_{\nu}dB_t.
$$
Substitute the above equation into our value function, 
divide one ``$dt$'' on both sides and then let $dt\to0$, then we can get the result. 
 According to the definition of the value function,  it's not difficult to verify  the boundary condition $V(q,\nu,T)= 0$ .

\subsection*{B2. Proof of Theorem 1}

Assume that the arrival rates take the exponential form:
$$
\lambda^b(\delta) = \lambda^a(\delta)=Ae^{-k\delta}
$$ 
for some constants $k$ and $A$,
 the optimal distances 
$\delta_t^{b,*}$ and $\delta_t^{a,*}$ are  then given by
$$
\left\{
\begin{array}{lll}
\delta_t^{b,*}=\frac{1}{k}+\beta-\left[V(q+1,\nu,t)-V(q,\nu,t)\right]\\
\delta_t^{a,*}=\frac{1}{k}-\beta-\left[V(q-1,\nu,t)-V(q,\nu,t)\right].
\end{array}
\right.
$$
Substituting the optimal distances into  
Eq. (\ref{PDE})  yields 
\begin{equation}\label{B2}
\displaystyle V_t+\theta(\alpha-\nu)V_\nu+\frac{1}{2}\xi^2\nu V_{\nu\nu}-\frac{\gamma}{2} q^2\nu+ \frac{A}{k}\left(e^{-k\delta_t^{a,*}}+e^{-k\delta_t^{b,*}}\right)=0
\end{equation}
with boundary condition $V(q,\nu,T)=0$.
Using Taylor's expansion,
\begin{equation}\label{AE}
V(q,\nu,t)=f^{(0)}(\nu,t)+f^{(1)}(\nu,t)q+f^{(2)}(\nu,t)q^2+\cdots+.
\end{equation}
Then
$$
\left\{
\begin{array}{l}
\delta_t^{b,*}\approx\frac{1}{k}+b-f^{(1)}(\nu,t)-f^{(2)}(\nu,t)(2q+1)\\
\delta_t^{a,*}\approx\frac{1}{k}-b+f^{(1)}(\nu,t)+f^{(2)}(\nu,t)(2q-1)
\end{array}
\right.
$$
and
$$
\delta_t^{b,*} + \delta_t^{a,*} \approx \frac{2}{k}-2f^{(2)}(\nu,t)
$$
which is independent of the inventory.
Adopte the method in \cite{Avellaneda08} and  take the first-order approximation of the order arrival term, 
\begin{equation}\label{AP}
\frac{A}{k}\left(e^{-k\delta_t^{b,*}} + e^{-k\delta_t^{a,*}}\right)
\approx \frac{A}{k}\left(2-k(\delta_t^{b,*}+\delta_t^{a,*})\right).
\end{equation}
Note that the linear term ($\delta_t^{a,*}+\delta_t^{b,*}$)  
does not depend on the inventory $q$. 
Therefore, if we substitute Eq. (\ref{AE}) and  Eq. (\ref{AP})
into Eq. (\ref{B2}) and grouping terms of $q$, then 
\begin{equation}\label{un1}
\left \{
\begin{array}{lll}
f_t^{(1)}+\theta(\alpha-\nu)f_\nu^{(1)}+\frac{1}{2}\xi^2\nu f_{\nu\nu}^{(1)}=0\\
f^{(1)}(\nu,T)=0.
\end{array}
\right.
\end{equation}
By the Feynman-Kac formula, $f^{(1)}(\nu,t)=0$.
Group terms in the coefficients of  $q^2$, then
$$
\left \{
\begin{array}{lll}
f_t^{(2)}+\theta(\alpha-\nu)f_\nu^{(2)}+\frac{1}{2}\xi^2\nu f_{\nu\nu}^{(2)}-\frac{\gamma}{2}\nu=0\\
f^{(2)}(\nu,T)=0
\end{array}
\right.
$$
whose solution can be directly obtained  using the Feynman-Kac formula 
\begin{equation}
\begin{array}{lll}
f^{(2)}(\nu,t)&=&\displaystyle E_t\left[ \int_t^T -\frac{1}{2}\gamma \nu_udu\right]\\
&=&-\displaystyle\frac{1}{2}\gamma \int_t^T E_t[ \nu_u]du\\
&=&-\displaystyle\frac{1}{2}\gamma \int_t^T \left[e^{-\theta(u-t)}\nu +\alpha \left(1-e^{-\theta(u-t)}\right)\right]du\\
&=&-\displaystyle\frac{\gamma}{2\theta}(\nu-\alpha)\left[1-e^{-\theta(T-t)}\right]-\frac{\gamma}{2}\alpha(T-t).
\end{array}
\end{equation}
The initial condition $V(0,v,t)=0$ tells  $f^{(0)}(\nu,t)=0$.
Thus, 
$$
\tilde{V}(q,\nu,t)\approx -\displaystyle\frac{\gamma q^2}{2\theta}(\nu-\alpha)\left[1-e^{-\theta(T-t)}\right]
-\frac{\gamma q^2}{2}\alpha(T-t)
$$
For the  quadratic polynomial asymptotic expansion of $V(q,\nu,t)$ in the inventory variable $q$ and   the linear approximation of the order arrival terms,
we obtain the same value function for
active dealers as in the case of inactive dealers.
The approximated optimal quotes are given by
$$
\left\{
\begin{array}{lll}
\hat{\delta}_t^{a,*}&=&\frac{1}{k}-\beta-\left(\frac{\gamma}{2\theta}(\nu-\alpha)\left[1-e^{-\theta(T-t)}\right]+\frac{\gamma}{2}\alpha(T-t)\right)(2q-1)\\
\hat{\delta}_t^{b,*}&=&\frac{1}{k}+\beta+\left(\frac{\gamma}{2\theta}(\nu-\alpha)\left[1-e^{-\theta(T-t)}\right]+\frac{\gamma}{2}\alpha(T-t)\right)(2q+1).
\end{array}
\right.
$$
We now analyze the difference between the approximate
and  exact optimal quotes under the stochastic volatility model.
Firstly, we discuss  the difference between the approximate and exact solutions of Eq. (\ref{B2}).
Suppose that
\begin{equation}
V(q,\nu,t)=w^q(t,\nu)-\displaystyle\frac{\gamma q^2}{2\theta}(\nu-\alpha)\left[1-e^{-\theta(T-t)}\right]-\frac{\gamma q^2}{2}\alpha(T-t)
\end{equation}
where $(w^q)_{q\in \mathbb{N}}$ is a family of functions in ${\cal C}^{1,2}$, and $\mathbb{N}$ is the set of natural numbers.
Substituting the above expression  into Eq.  (\ref{B2}),
we obtain
\begin{equation}
\left \{
\begin{array}{lll}
w_t^q+\theta(\alpha-\nu)w_\nu^q+\frac{1}{2}\xi^2\nu w_{\nu\nu}^q+g^q(\nu,t)=0\\
w^q(\nu,T)=0
\end{array}
\right.
\end{equation}
where
\begin{equation}
\begin{array}{lll}
 g^q(\nu,t)&=&  \frac{A}{k}(e^{-k\delta_t^a}+e^{-k\delta_t^b})\\
 &=& \frac{A}{k}\Big[e^{\left(-k(w^q-w^{q-1}+\frac{1}{k}-b+\left(\frac{\gamma}{2\theta}(\nu-a)(1-e^{-\theta(T-t)})+\frac{\gamma}{2}\alpha(T-t)\right)(-2q+1)\right)}\\
 &+&  e^{\left(-k(w^q-w^{q+1}+\frac{1}{k}+b+\left(\frac{\gamma}{2\theta}(\nu-a)(1-e^{-\theta(T-t)})+\frac{\gamma}{2}\alpha(T-t)\right)(2q+1)\right)}\Big]
 \end{array}
\end{equation}
is a family of non-negative functions indexed by $q$.
Then the Feynman-Kac formula implies that the solution can be written as
the conditional expectation:
$$
w^q(\nu,t)=E_t\left[\int_t^Tg^q(\nu_r,r)dr \right].
$$
Since $g^q(\nu_r,r)$ is nonnegative
and $(\delta_t^a,\delta_t^b)$ is supposed to be bounded from below,
  there exists some positive constant $c$  such that
$0\le g^q(\nu,t) \le c$. Then
$$
\tilde{J}(q,v,t) \le J(q,\nu,t) \le c(T-t)+\tilde{J}(q,v,t)
$$
and
$$
\left\{
\begin{array}{lll}
\underbrace{\delta_t^{a,*}}_{exact}-\underbrace{\hat{\delta}_t^{a,*}}_{approx}&=&w^q-w^{q-1} = E_t \left[ \int_t^T(g^q-g^{q-1})(\nu_r,r)dr \right]\\
\underbrace{\delta_t^{b,*}}_{exact}-\underbrace{\hat{\delta}_t^{b,*}}_{approx}&=&w^q-w^{q+1} = E_t \left[ \int_t^T(g^q-g^{q+1})(\nu_r,r)dr \right].\\
\end{array}
\right.
$$
Note that
$$
g^q(\nu,t)=\frac{A}{k}(e^{-k\delta_t^{a,*}}+e^{-k\delta_t^{b,*}})
$$
where $\delta_t^{a,*}$ and $ \delta_t^{b,*}$ are relatively small
and $\delta_t^{a,*}+\delta_t^{b,*}$ is independent of $q$. 
Thus, we have
$$
g^q(\nu,t)\approx\frac{A}{k}\left(2-k(\delta_t^{a,*}+\delta_t^{b,*})+\frac{k^2}{2}\left((\delta_t^{a,*})^2+(\delta_t^{b,*})^2\right)\right).
$$
The differences between the exact and the approximate ask and bid quotes, $|\delta_t^{a,*}-\hat{\delta}_t^{a,*}|$ and $|\delta_t^{b,*}-\hat{\delta}_t^{b,*}|$, 
can also be  very small.

\subsection*{C.}
With the intensity function 
$$
\lambda^a(\delta)=\lambda^b(\delta)=A\exp(-k\delta)
$$
and similar method to the proof of Theorem 1,  we have 
\begin{equation}\label{e1}
V(q,\nu,t)=V^{(0)}(\nu,t)+V^{(1)}(\nu,t)q+V^{(2)}(\nu,t)q^2+\cdots+
\end{equation}
\begin{equation}\label{e2}
\frac{A}{k}\left(e^{-k\delta_t^{*,a}}+e^{-k\delta_t^{*,b}}\right)\approx
\frac{A}{k}\left(2-k(\delta_t^{*,a}+\delta_t^{*,b})\right).
\end{equation}
Thus,
\begin{equation}\label{e3}
\delta_t^{*,a}+\delta_t^{*,b}\approx\frac{2}{k}+ \gamma \eta^2(q^2+1)-2V^{(2)}(\nu,t).
\end{equation}
Different from the first model, the ask-bid spread is related to the inventory variable $q$. 
If we substitute Eq.s (\ref{e1}), (\ref{e2}) and (\ref{e3}) into 
Eq. (\ref{PDEM}) and group the coefficients of the term $q$, we obtain
$$
\left\{
\begin{array}{lll}
\displaystyle V_t^{(1)}+\theta(\alpha-\nu)V_\nu^{(1)}+\frac{1}{2}\xi^2\nu V_{\nu\nu}^{(1)}=0\\
V^{(1)}(q,\nu,T)=0\\
\end{array}
\right.
$$
whose solution is $V^{(1)}(\nu,t)=0$. Grouping the coefficients of the term $q^2$ yields
$$
\left\{
\begin{array}{lll}
\displaystyle V_t^{(2)}+\theta(\alpha-\nu)V_\nu^{(2)}+\frac{1}{2}\xi^2\nu V_{\nu\nu}^{(2)}
-\frac{\gamma}{2}\nu-A\gamma \eta^2=0\\
V^{(2)}(q,\nu,T)=0.\\
\end{array}
\right.
$$
According to the Feynman-Kac formula,  
$$
\begin{array}{lll}
\displaystyle V^{(2)}(\nu,t)&=&\displaystyle -\frac{\gamma}{2}E_t\int_t^T\left(\nu_u
+2A \eta^2\right)du\\
&=&\displaystyle -\frac{\gamma}{2}\int_t^T\left(E_t[\nu_u]
+2A \eta^2\right)du\\
&=&\displaystyle-\frac{\gamma}{2}\left(\frac{1}{\theta}(\nu-\alpha)\left[
1-e^{-\theta(T-t)}\right]+\alpha(T-t)+2A\eta^2(T-t)\right).\\
\end{array}
$$
Then, we can get the approximative optimal quotes
$$
\left\{
\begin{array}{lll}
\displaystyle\hat{ \delta}_t^{*,a}=\frac{1}{k} - \beta+\frac{\gamma \eta^2}{2}(q-1)^2-\frac{\gamma}{2}\left(\frac{1}{\theta}(\nu-\alpha)\left[
1-e^{-\theta(T-t)}\right]+\alpha(T-t)+2A\eta^2(T-t)\right)(2q-1)\\
\displaystyle \hat{\delta}_t^{*,b}=\frac{1}{k} + \beta +\frac{\gamma \eta^2}{2}(q+1)^2+\frac{\gamma}{2}\left(\frac{1}{\theta}(\nu-\alpha)\left[
1-e^{-\theta(T-t)}\right]+\alpha(T-t)+2A\eta^2(T-t)\right)(2q+1).\\
\end{array}
\right.
$$

\subsection*{D1. Proof of Proposition 3.}

Under the actual probability measure $\mathbf{P}$, we have 
$$
\left\{
\begin{array}{lll}
dS_t=\sqrt{\nu_t}dW_t\\
d\nu_t=\theta(\alpha-\nu_t)dt+\xi\sqrt{\nu_t}\Big(\rho dW_t+\sqrt{1-\rho^2}d\tilde{W_t}\Big)\\
dW_t\cdot d\tilde{W_t}=0.
\end{array}
\right.
$$
In our setting, the interest rate $r=0$. Thus, under the risk-neutral measure $\mathbf{Q}$, the price of the stock price risk equals zero, i.e. 
$$ 
dW_t^{\mathbf{Q}}=dW_t\quad {\it and}\quad d\tilde{W_t}^{\mathbf{Q}}=d\tilde{W_t}+\eta^{\nu}dt
$$
where, $\eta^{\nu}$ is the price of volatility risk not related to stock returns. Thus,
$$
\left\{
\begin{array}{lll}
dS_t=\sqrt{\nu_t}dW_t^{\mathbf{Q}}\\
d\nu_t=\Big[\theta(\alpha-\nu_t)-\xi\sqrt{\nu_t}\sqrt{1-\rho^2}\eta^{\nu}\Big]dt+\xi\sqrt{\nu_t}\Big(\rho dW_t^{\mathbf{Q}}+\sqrt{1-\rho^2}d\tilde{W_t}^{\mathbf{Q}}\Big)\\
dW_t^{\mathbf{Q}}\cdot d\tilde{W_t}^{\mathbf{Q}}=0.
\end{array}
\right.
$$
The price of a European call option can then be computed as follows:
$$
C(s,\nu,t)=E_t^{\mathbf{Q}}[(S_T-K)^+].
$$
Under the risk-neutral measure, all the discounted asset prices are martingale, 
the drift term of
$$
dC(s,\nu,t)=\Big[C_t+\frac{1}{2}\nu_tC_{ss}+\xi\rho\nu_tC_{s\nu}+\frac{1}{2}\xi^2\nu_tC_{\nu\nu}+\big[\theta(\alpha-\nu_t)-\xi\sqrt{\nu_t}\sqrt{1-\rho^2}\eta^{\nu}\big]C_{\nu}\Big] dt + \ldots +.
$$
is identical to zero. Consequently,
$$
C_t+\frac{1}{2}\nu_tC_{ss}+\xi\rho\nu_tC_{s\nu}+\frac{1}{2}\xi^2\nu_tC_{\nu\nu}+\big[\theta(\alpha-\nu_t)-\xi\sqrt{\nu_t}\sqrt{1-\rho^2}\eta^{\nu}\big]C_{\nu}=0
$$
with terminal condition $C(s,\nu,T)=(s-K)^+$.

\subsection*{D2. Proof of Proposition 4.}

Suppose we have $N(N>2)$ assets with two-factor structure in their returns:
$$
R_t^i=a_i+b_i^1Y_1(t)+b_i^2Y_2(t)
$$
and the risk-free interest rate is $r$. At time $t$, consider any portfolio with fraction $\theta_i$ in each asset $i$ that has zero exposure to both factors:
$$
\left\{
\begin{array}{lll}
b_1^1\theta_1+b_2^1\theta_2+\cdots +b_N^1\theta_N&=&0\\
b_1^2\theta_1+b_2^2\theta_2+\cdots +b_N^2\theta_N&=&0.
\end{array}
\right.
$$
According to the standard APT (Arbitrage Pricing Theorem) \cite{APT}, 
this portfolio must have zero expected excess return to avoid arbitrage
$$
\theta_1(E_t[R_t^1]-r)+\theta_2(E_t[R_t^2]-r)+\ldots+\theta_N(E_t[R_t^N]-r)=0.
$$
Restating this in vector form, any vector orthogonal to $(b_1^1,b_2^1,\ldots,b_N^1)$ and $(b_1^2,b_2^2,\ldots,b_N^2)$ 
must be orthogonal to $(E_t[R_t^1]-r, E_t[R_t^2]-r,\ldots, E_t[R_t^N]-r)$, 
which means that the third vector is spanned by the first two: there exist constants (prices of risk) $(\lambda_t^1,\lambda_t^2)$ such that
$$
E_t[R_t^i]-r=\lambda_t^1 b_i^1+\lambda_t^2 b_i^2, \quad i=1,2,\ldots,N.
$$
As for the option pricing in the Heston's Model, using Ito's lemma, option price, $C(s,\nu,t)$, satisfies
$$
dC(s,\nu,t)=A(t) dt+C_sdS_t+C_{\nu}d\nu_t
$$
where, $A(t)$ consists of  all the terms of  $dt$.
Compare the above to the APT argument, there exist  $\lambda_t^s$ and $\lambda_t^{\nu}$ such that(we work with price changes instead of returns)
\begin{equation}\label{APT}
E[dC(s,\nu,t)-rC(s,\nu,t)dt]=C_s\lambda_t^sdt+C_{\nu}\lambda_t^{\nu}dt.
\end{equation}
The APT pricing equation, applied to the stock, implies that 
$$
E[dS_t-rS_tdt]=-rS_tdt=\lambda_t^sdt.
$$
Thus, $\lambda_t^s=-rS_t$,  denoting the market price of risk of the stock.    Correspondingly,  $\lambda_t^{\nu}$ is the price of volatility risk, which determines the risk premium on any investment with exposure to $d\nu_t$.

Writing down the pricing equation (\ref{APT}) explicitly with the Ito's lemma yields,
$$
C_t+\frac{1}{2}\nu_tC_{ss}+\xi\rho\nu_tC_{sv}+\frac{1}{2}\xi^2\nu_tC_{\nu\nu}+\theta(\alpha-\nu_t)C_{\nu}-rC=-C_srS_t+C_{\nu}\lambda_t^{\nu}.
$$
As long as we assume that the price of volatility risk is of the form
$$
\lambda_t^{\nu}=\lambda^{\nu}(t,S_t,\nu_t)
$$
the assumed functional form for option prices is justified and we obtain an arbitrage-free option pricing model. 
Since market model under Heston's stochastic volatility assumption is incomplete, there exists more than one risk-neutral probability measure, 
so does the option price.  Letting $r=0$, we can get Eq. (\ref{C2}).

\subsection*{D3. Proof of Theorem 3.}

The order arrival terms are highly nonlinear and may depend on the inventory.
Directly applying the method in Avellaneda and Stoikov (2008) \cite{Avellaneda08} and using an asymptotic expansion to approximate $V(s,\nu,q_t^s,q_t^o,t)$ as a quadratic polynomial in the inventory variables:
$q_t^s$ and $q_t^o$. We have
\begin{equation}\label{J1}
\begin{array}{lll}
V(s,\nu,q_t^s,q_t^o,t)&=&a(s,\nu,t)+b_1(s,\nu,t)q_t^s+b_2(s,\nu,t)q_t^o\\
&&+c_1(s,\nu,t)(q_t^s)^2+c_2(s,\nu,t)q_t^sq_t^o+c_3(s,\nu,t)(q_t^o)^2+\ldots+.
\end{array}
\end{equation} 
One can derive an approximate solution,
$\tilde{V}(s,\nu,q_t^s,q_t^o,t)$,
for the unknown function $V(s,\nu,q_t^s,q_t^o,t)$.
For the P.D.E. in Eq. (\ref{PDEO}),
taking the first order approximation of the arrival term
$$
\begin{array}{lll}
\frac{A}{k}\left(e^{-k\delta_{t,*}^{a,s}}+e^{-k\delta_{t,*}^{b,s}}
+e^{-k\delta_{t,*}^{a,o}}+e^{-k\delta_{t,*}^{b,o}}\right)&=&\frac{A}{k}\left(4-k(\delta_{t,*}^{a,s}+\delta_{t,*}^{b,s}
+\delta_{t,*}^{a,o}+\delta_{t,*}^{b,o})+\cdots+ \right)
\end{array}
$$
and notice that
$$
\begin{array}{lll}
\delta_{t,*}^{a,s}+\delta_{t,*}^{b,s}+\delta_{t,*}^{a,o}+\delta_{t,*}^{b,o}&=&\frac{4}{k}+4V(s,\nu,q_t^s,q_t^o,t)-
V(s,\nu,q_t^s-1,q_t^o,t)-V(s,\nu,q_t^s+1,q_t^o,t)\\
&&-V(s,\nu,q_t^s,q_t^o-1,t)-V(s,\nu,q_t^s,q_t^o+1,t)\\
&\approx&\frac{4}{k}-2\Big(c_1(s,\nu,t)+c_3(s,\nu,t)\Big)\\
\end{array}
$$
which is independent of inventories: $q_t^s$ and $q_t^o$, so is the linear term in the arrival term. 
Grouping the terms of order $q_t^s$, we obtain
\begin{equation}\label{ORq1}
\left\{
\begin{array}{lll}
\displaystyle (b_1)_t+\theta(\alpha-\nu_t)(b_1)_\nu+\frac{1}{2}\nu_t(b_1)_{ss}+\rho\xi \nu_t(b_1)_{s\nu}+\frac{1}{2}\xi^2 \nu_t(b_1)_{\nu\nu}=0\\
b_1(s,\nu,T)=0.\\
\end{array}
\right.
\end{equation}
 By the Feynman-Kac formula  $ b_1(s,\nu,t)=0 $.
Grouping terms of order $q_t^o$, 
\begin{equation}\label{ORq2}
\left\{
\begin{array}{lll}
\displaystyle (b_2)_t+\theta(\alpha-\nu_t)(b_2)_\nu+\frac{1}{2}\nu_t(b_2)_{ss}+\rho\xi \nu_t(b_2)_{s\nu}+\frac{1}{2}\xi^2 \nu_t (b_2)_{\nu\nu}=0\\
b_2(s,\nu,T)=0.
\end{array}
\right.
\end{equation}
Therefore, $b_2(s,\nu,t)=0$.
Grouping terms of order $(q_t^s)^2$,  
\begin{equation}\label{ORq3}
\left\{
\begin{array}{lll}
\displaystyle (c_1)_t+\theta(\alpha-\nu_t)(c_1)_\nu+\frac{1}{2}\nu_t(c_1)_{ss}+\rho\xi \nu_t(c_1)_{s\nu}+\frac{1}{2}\xi^2 \nu_t (c_1)_{\nu\nu}
-\frac{\gamma}{2}\nu_t=0\\
c_1(s,\nu,T)=0\\
\end{array}
\right.
\end{equation}
whose solution is again obtained from the Feynman-Kac formula as follows:
$$
\begin{array}{lll}
c_1(s,\nu,t)&=&-\frac{\gamma}{2} E_t \left[\int_t^T\nu_udu \right]\\
&=&-\frac{\gamma}{2}\int_t^TE_t[\nu_u]du\\
&=&-\displaystyle\frac{\gamma}{2\theta}(\nu-\alpha)
\left[1-e^{-\theta(T-t)}\right]-\frac{\gamma}{2}\alpha(T-t).
\end{array}
$$
Grouping terms of $q_t^sq_t^o$, we obtain
\begin{equation}\label{ORq4}
\left\{
\begin{array}{lll}
\displaystyle (c_2)_t+\theta(\alpha-\nu_t)(c_2)_\nu+\frac{1}{2}\nu_t(c_2)_{ss}+\rho\xi \nu_t(c_2)_{s\nu}+\frac{1}{2}\xi^2 \nu_t (c_2)_{\nu\nu}
-\gamma\nu_t\Big(\Delta_t +\rho\xi  C_\nu\Big)=0\\
c_2(s,\nu,T)=0.
\end{array}
\right.
\end{equation}
Then applying the Feynman-Kac formula yields
$$
c_2(s,\nu,t)=-\gamma E_t \left[ \int_t^T\nu_u\Big(\Delta_u +\rho\xi  C_\nu(S_u,\nu_u,u)\Big)du\right] = H_1(s,\nu,t).
$$
Grouping terms of order $(q_t^o)^2$, we obtain
\begin{equation}\label{ORq5}
\left\{
\begin{array}{lll}
\displaystyle (c_3)_t+\theta(\alpha-\nu_t)(c_3)_\nu+\frac{1}{2}\nu_t(c_3)_{ss}+\rho\xi \nu_t(c_3)_{s\nu}+\frac{1}{2}\xi^2 \nu_t (c_3)_{\nu\nu}
-\frac{\gamma}{2}\nu_t\Big(\Delta_t^2
+2\rho\xi\Delta_t C_\nu
+\xi^2C_\nu^2\Big)=0\\
c_3(s,\nu,T)=0.
\end{array}
\right.
\end{equation}
Similarly using the Feynman-Kac formula
$$
c_3(s,\nu,t)=-\frac{\gamma}{2} E_t \left[\int_t^T\nu_u\Big(\Delta_u^2
+2\rho\xi\Delta_u C_\nu(S_u,\nu_u,u)
+\xi^2 C_\nu^2(S_u,\nu_u,u)\Big)du \right] = H_2(s,\nu,t).
$$
Furthermore, the initial condition ensures that $a(s,\nu,t)=0$.  
Suppose that
$$
V(s,\nu,q_t^s,q_t^o,t)= w(s,\nu,q_t^s,q_t^o,t)+ \tilde{V}(s,\nu,q_t^s,q_t^o,t)
$$
then
the unknown function $w(s,\nu,q_t^s,q_t^o,t)$ satisfies
$$
\left\{
\begin{array}{lll}
\displaystyle w_t+\theta(\alpha-\nu_t)w_\nu+\frac{1}{2}\nu_tw_{ss}+\xi\rho \nu_tw_{s\nu}+\frac{1}{2}\xi^2 \nu_t w_{\nu\nu} +\frac{A}{k}\left(e^{-k\delta_{t,*}^{a,s}}+e^{-k\delta_{t,*}^{b,s}}+e^{-k\delta_{t,*}^{a,o}}+e^{-k\delta_{t,*}^{b,o}}\right)=0\\
w(s,\nu,q_t^s,q_T^o,T)=0.
\end{array}
\right.
$$
The Fernman-Kac formula then gives
$$
w(s,\nu,q_t^s,q_t^o,t)
=\frac{A}{k}E_t \left[\int_t^T\left(e^{-k\delta_{u,*}^{a,s}}+e^{-k\delta_{u,*}^{b,s}}+e^{-k\delta_{u,*}^{a,o}}+e^{-k\delta_{u,*}^{b,o}}\right)du \right] \ge 0.
$$

\subsection*{D4. Proof of Theorem 4.}

The optimal quotes are obtained through a two-step procedure.
First, solve the P.D.E. in Eq. (\ref{PDEOO}) in order to obtain the unknown function $V(s,\nu,q_t^o,t)$. Second, substitute $V(s,\nu,q_t^o,t)$ into Eq. (\ref{controlOO}) and obtain the optimal premiums $\delta_{t,*}^{a,o}$ and $\delta_{t,*}^{b,o}$. 
Our first task is to solve the unknown function $V(s,\nu,q_t^o,t)$ using the exponential arrival rates.

We directly apply Avellaneda and Stoikov's method \cite{Avellaneda08} and using the fact that $V(s,\nu,q_t^o,t)$
is an approximate quadratic polynomial in the inventory variable,
$$
V(s,\nu,q_t^o,t)=V^{(0)}(s,\nu,t)+V^{(1)}(s,\nu,t)q_t^o+V^{(2)}(s,\nu,t)(q_t^o)^2+\ldots+.
$$
We can first derive an approximate solution
$\tilde{V}(s,\nu,q_t^o,t)$.
For the P.D.E. (\ref{PDEOO}),
taking the first order approximation of the arrival term
$$
\frac{A}{k}\left(e^{-k\delta_{t,*}^{a,o}}+e^{-k\delta_{t,*}^{b,o}}\right)=\frac{A}{k}\left(2-k(\delta_{t,*}^{a,o}+\delta_{t,*}^{b,o})+\cdots + \right)
$$
and notice that
$$
\begin{array}{lll}
\delta_{t,*}^{a,o}+\delta_{t,*}^{b,o}&=&\frac{2}{k}+2V(s,\nu,q_t^o,t)-
V(s,\nu,q_t^o-1,t)-V(s,\nu,q_t^o+1,t)\\
&\approx&\frac{2}{k}-2V^{(2)}(s,\nu,t)\\
\end{array}
$$
which is independent of inventory, so is the linear term in the arrival term.  Grouping terms of order $q_t^o$, we obtain
\begin{equation}\label{ORq}
\left\{
\begin{array}{lll}
\displaystyle V_t^{(1)}+\theta(\alpha-\nu_t)V_\nu^{(1)}+\frac{1}{2}\nu_tV_{ss}^{(1)}+\xi\rho \nu_tV_{s\nu}^{(1)}+\frac{1}{2}\xi^2 \nu_t V_{\nu\nu}^{(1)}=0\\
V^{(1)}(s,\nu,T)=0.
\end{array}
\right.
\end{equation}
By the Feynman-Kac formula, $V^{(1)}(s,\nu,t)=0$. 
Grouping terms of order $(q_t^o)^2$, we obtain
\begin{equation}\label{ORqq}
\left\{
\begin{array}{lll}
\displaystyle V_t^{(2)}+\theta(\alpha-\nu_t)V_\nu^{(2)}+\frac{1}{2}\nu_tV_{ss}^{(2)}+\xi\rho \nu_tV_{s\nu}^{(2)}+\frac{1}{2}\xi^2 \nu_t V_{\nu\nu}^{(2)}
-\frac{\gamma}{2}
\nu_t\xi^2C_v^2=0\\
V^{(2)}(s,\nu,T)=0.\\
\end{array}
\right.
\end{equation}
The Feynman-Kac formula implies
$$
V^{(2)}(s,\nu,t)
=-\frac{\gamma}{2}\xi^2E_t\left[\int_t^T\nu_uC_\nu^2(S_u,\nu_u,u)du=M(s,\nu,t)\right].
$$
Moreover, the initial condition ensures that $V^{(0)}(s,\nu,t)=0$.
Suppose that
$$
V( s,\nu,q_t^o,t)= w(s,\nu,q_t^o,t)+ \tilde{V}(s,\nu,q_t^o,t)
$$
then
the unknown function $w(s,\nu,q_t^o,t)$ satisfies
$$
\left\{
\begin{array}{lll}
\displaystyle w_t+\theta(\alpha-\nu_t)w_\nu+\frac{1}{2}\nu_tw_{ss}+\xi\rho \nu_tw_{s \nu}+\frac{1}{2}\xi^2 \nu_t w_{\nu\nu} +\frac{A}{k}\left(e^{-k\delta_{t,*}^{a,o}}+e^{-k\delta_{t,*}^{b,o}}\right)=0\\
w(s,\nu,q_T^o,T)=0.\\
\end{array}
\right.
$$
The Fernman-Kac formula then suggest that
$$
w(s,\nu,q_t^o,t)=\frac{A}{k}E_t
\left[\int_t^T\left(e^{-k\delta_{u,*}^{a,o}}+e^{-k\delta_{u,*}^{b,o}}\right)du\right] \ge 0.
$$

\subsection*{Tables and Figures}
\begin{table}[!hbtp]\label{table 1}
\caption{1000 simulations without market impact with initial inventory $q_0=0$ }
\centering
\begin{tabular}{lccccc}
\hline
Strategy& Average Spread &Profit & Std (Profit) & $q_T$ & Std ($q_T$)\\
\hline
Inventory &1.53&64.68&6.70&-0.80&3.00\\
Symmetric &1.53&68.39&14.36&0.06&8.51\\
\hline
\end{tabular}
\end{table}

\begin{table}[!hbtp]\label{t2}
\caption{1000 simulations with market impact with initial inventory $q_0=0$ }
\centering
\begin{tabular}{lccccc}
\hline
Strategy& Average Spread &Profit & Std (Profit) & $q_T$ & Std ($q_T$)\\
\hline
Inventory &1.65& 62.56& 6.27&-0.55&2.65\\
Symmetric &1.65&67.63&12.75&-0.03&7.89\\
\hline
\end{tabular}
\end{table}

\begin{figure}[H]\label{fig1}
\centering
\includegraphics[width=6.0in,height=3.0in]{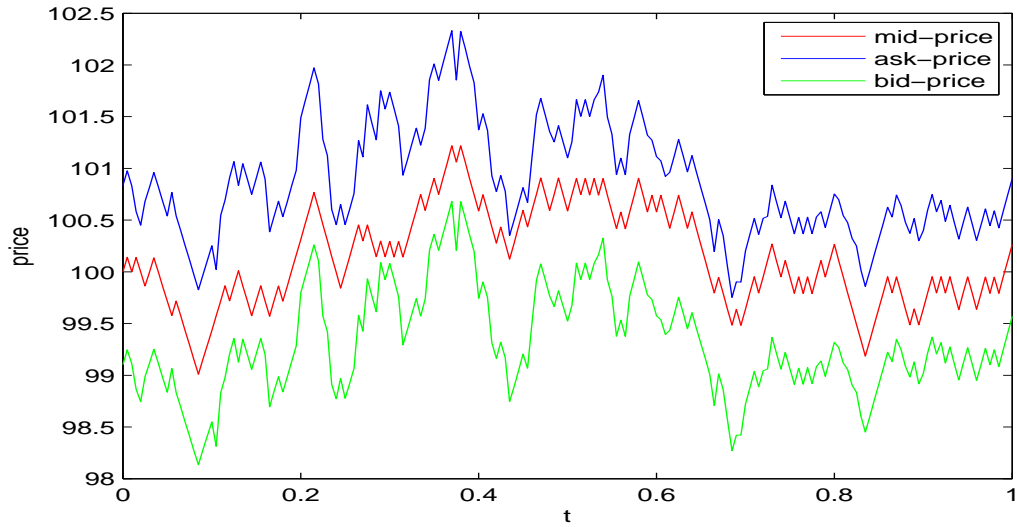}
\caption{The mid-price and the optimal bid-ask quotes without market impact}
\end{figure}

\begin{figure}[H]\label{fig2}
\centering
\subfigure[The cumulated revenues ]{
\includegraphics[width=3.0in,height=3.0in]{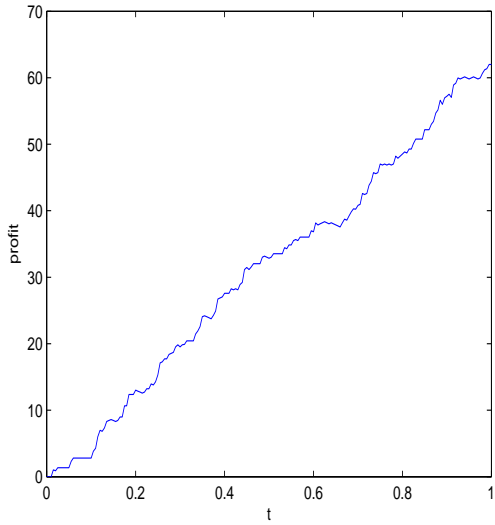} 
}
\subfigure[The ask-bid spread]{ 
\includegraphics[width=3.0in,height=3.0in]{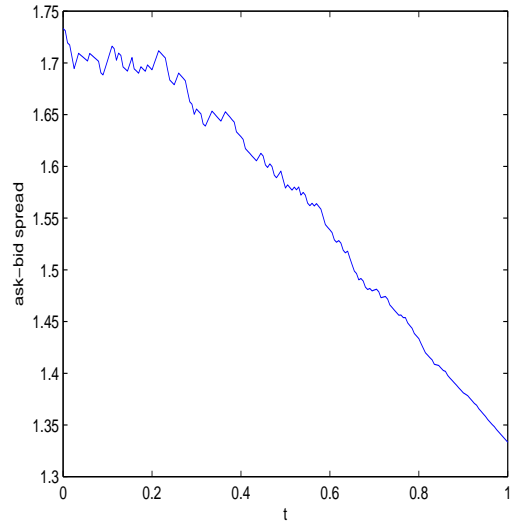} 
}
\caption{The trend charts of the cumulated revenues and the ask-bid spread  without market impact}
\end{figure}

\begin{figure}[H]\label{fig3}
\centering
\subfigure[Trading curve with  $q_0=6, \beta=0.03$.]{
\includegraphics[width=3.0in,height=3 in]{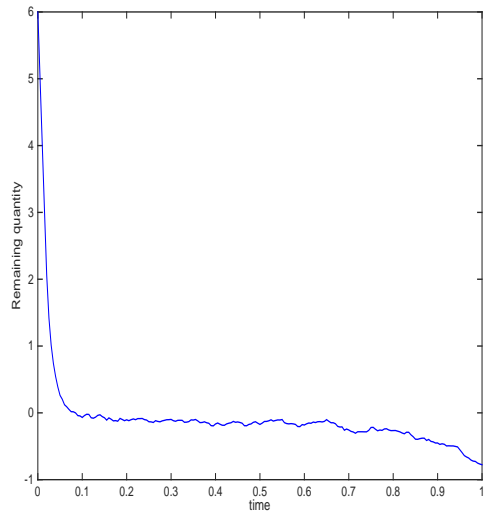} 
}
\subfigure[Trading curve with   $q_0=-6, \beta=0.03$.]{ 
\includegraphics[width=3.0in,height=3 in]{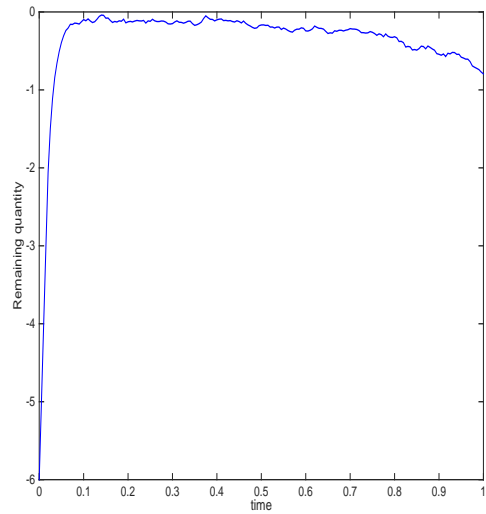} 
}
\caption{Trading curves under the setting without market impact  }
\end{figure}

\begin{figure}[H]\label{fig4}
\centering
\includegraphics[width=6.0 in,height=3.0in]{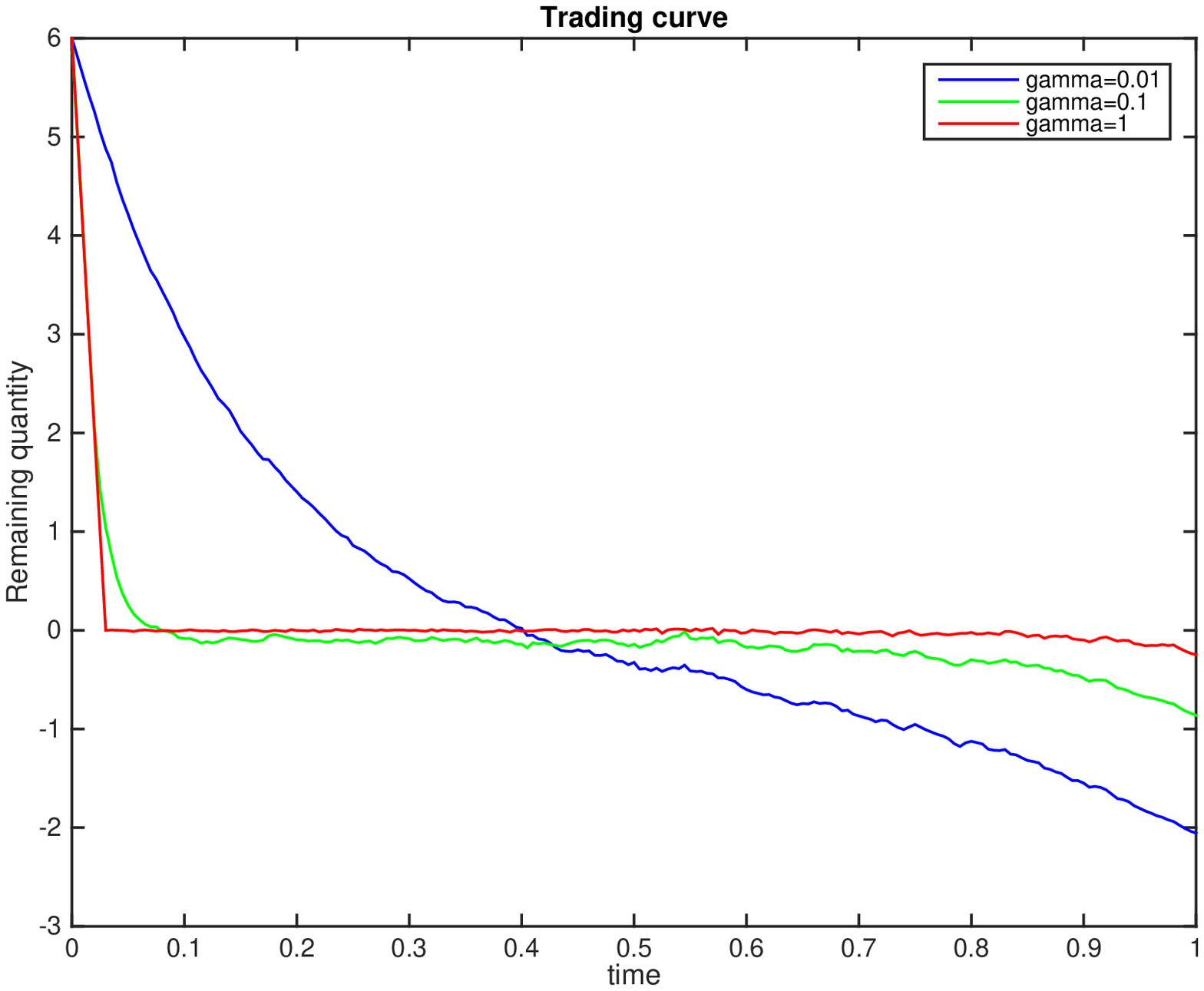}
\caption{Trading curves without market impact with initial inventory $q_0=6, \beta=0.03$, 
corresponding to 
Case (1) $\gamma=0.01$,
Case (2) $\gamma=0.10$, 
Case (3) $\gamma=1.00$.}
\end{figure}

 \begin{figure}[H]\label{fig5}
\centering
\includegraphics[width=6.0in,height=3.0in]{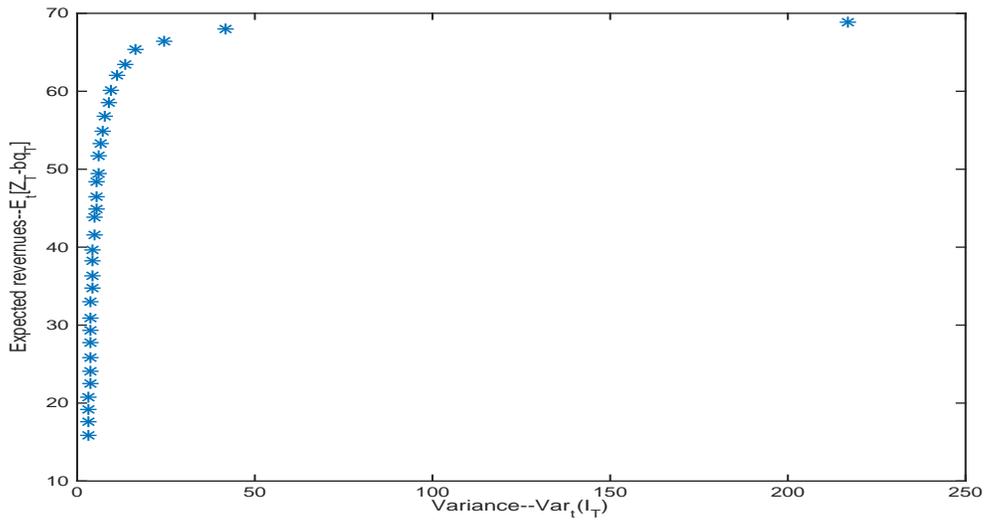}
\caption{Efficient frontier under the setting of without market impact : tradeoff between the cumulated revenues and inventory variance .}
\end{figure}

\begin{figure}[H]\label{fig6}
\centering
\includegraphics[width=6.0in,height=3.0in]{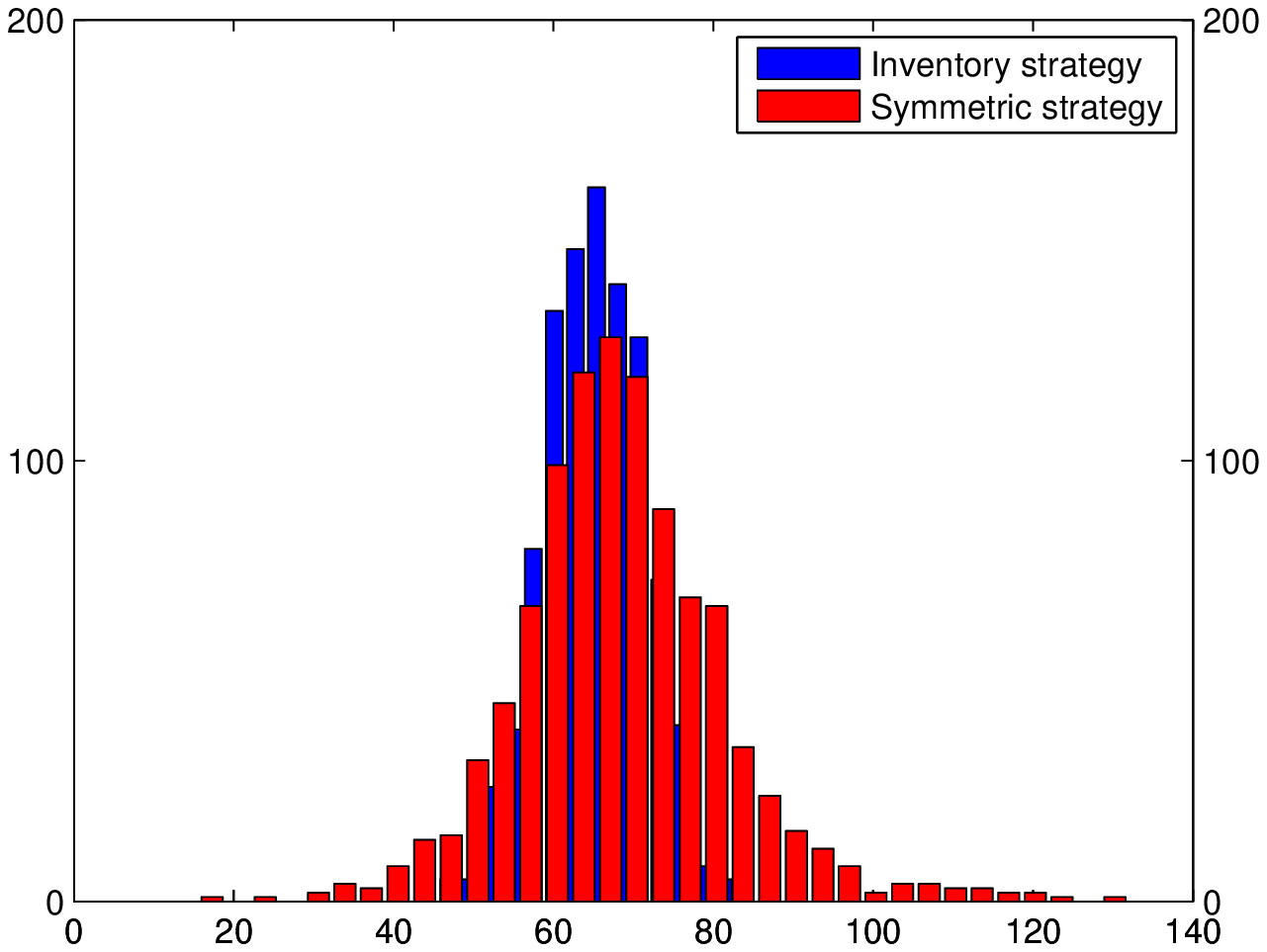}
\caption{A Comparison of two different strategies under the setting without market impact.}
\end{figure}

\begin{figure}[H]\label{figb7}
\centering
\includegraphics[width=6.0in,height=3.0in]{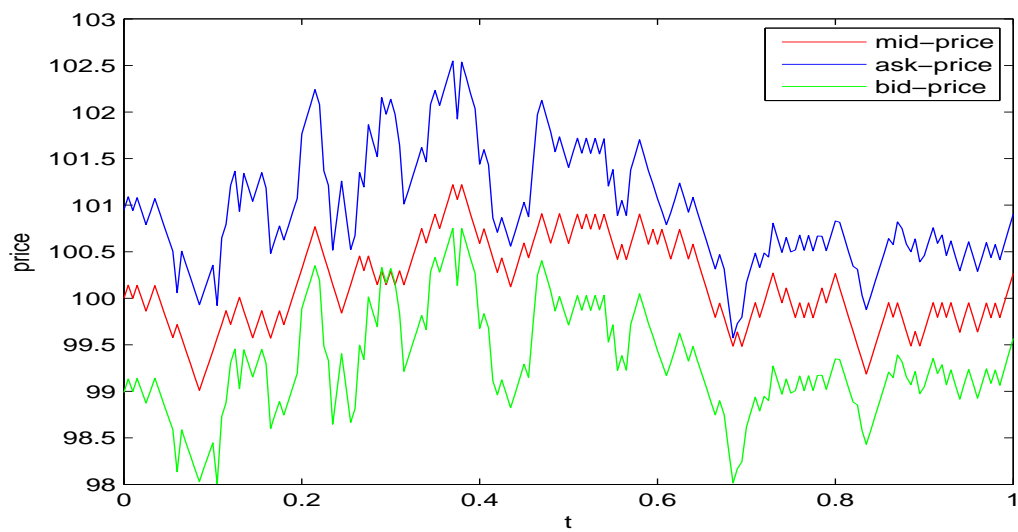}
\caption{The performance of mid-price and the optimal bid-ask quotes in the model taking the effect of market impact into consideration}
\end{figure}

\begin{figure}[H]\label{figb8}
\centering
\subfigure[The cumulated revenues]{
\includegraphics[width=3.0in,height=3. in]{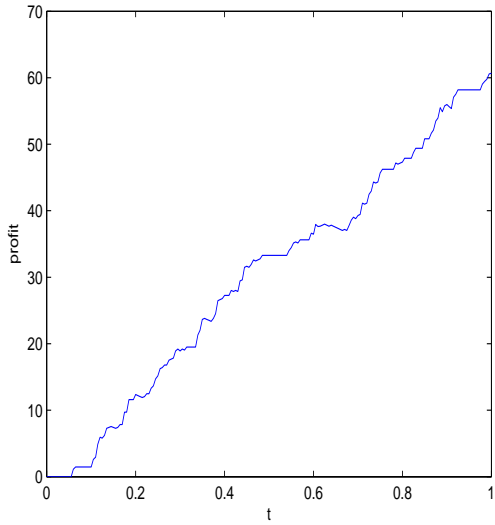} 
}
\subfigure[The ask-bid spread]{ 
\includegraphics[width=3.0in,height=3. in]{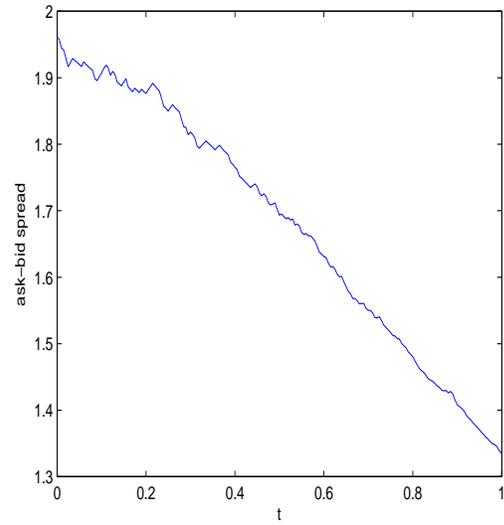} 
}
\caption{The Trend charts of the cumulated revenues and the ask-bid spread in a model with market impact. }
\end{figure}

\begin{figure}[H]\label{fig9}
\centering
\subfigure[Trading curve with  $q_0=6, \beta=0.03$]
{\includegraphics[width=3.0in,height=3. in]{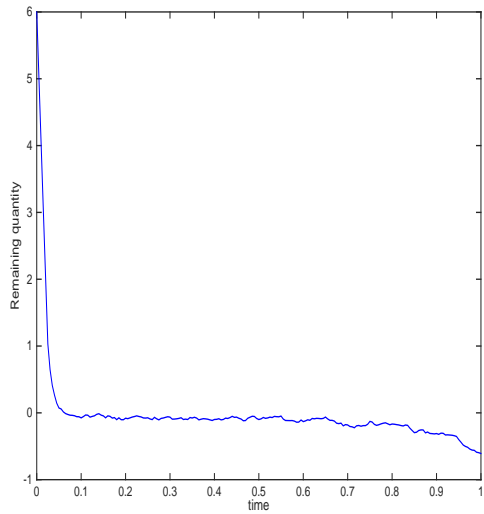} }
\subfigure[Trading curve with  $q_0=-6, \beta=0.03$]{ 
\includegraphics[width=3.0in,height=3. in]{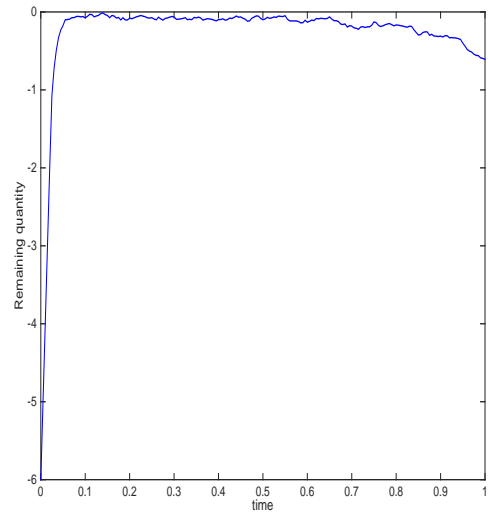} 
}
\caption{Trading curves under the setting taking the effect of  market impact into consideration. }
\end{figure}

\begin{figure}[H]\label{fig10}
\centering
\includegraphics[width=6.0in,height=3.0in]{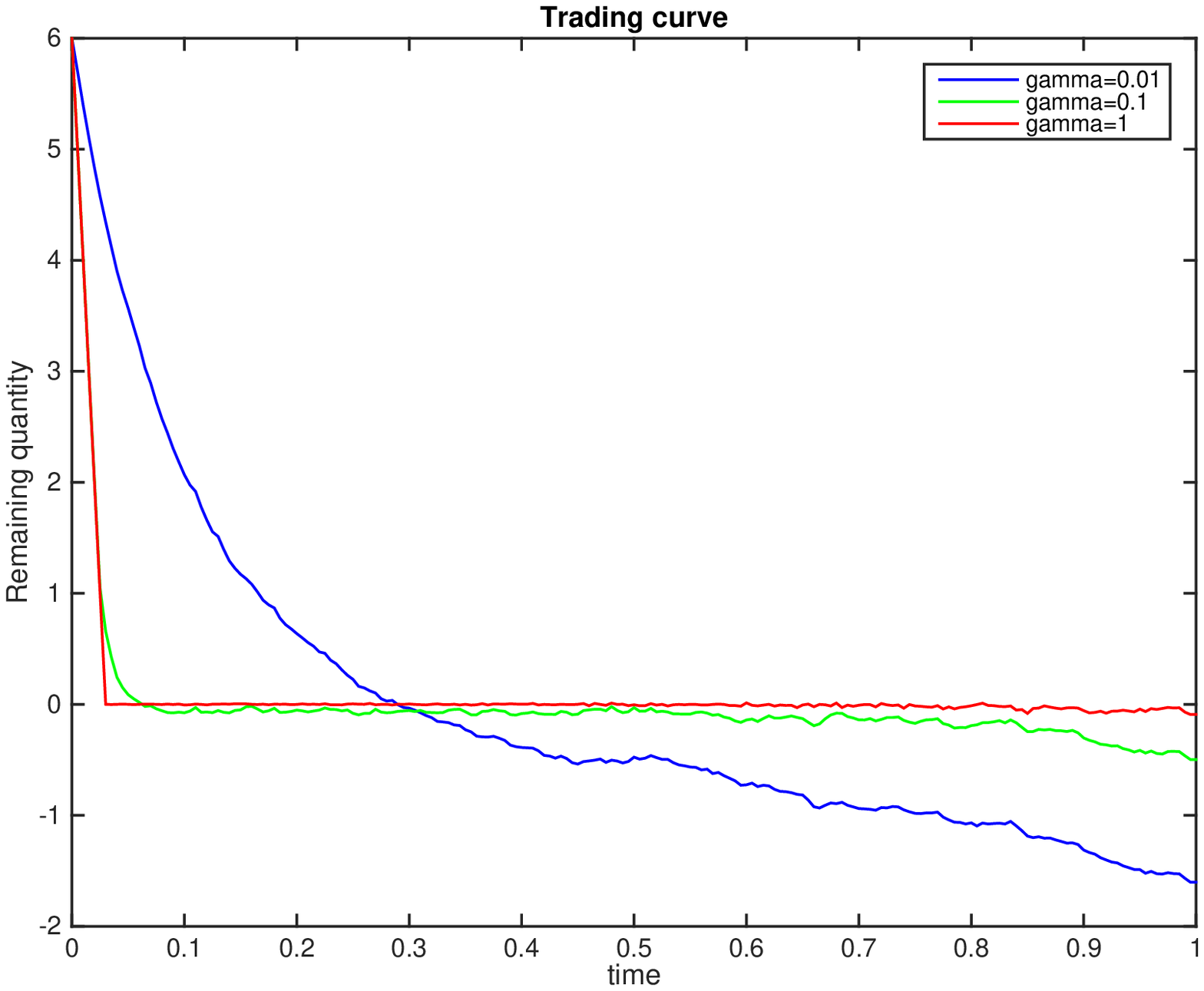}
\caption{Trading curves in a market with market impact, with initial inventory $q_0=6, \beta=0.03$, corresponding to Case (1) $\gamma=0.01$, Case (2) $\gamma=0.10$, Case (3) $\gamma=1.00$.}
\end{figure}

\begin{figure}[H]\label{fig11}
\centering
\includegraphics[width=6.0in,height=3.0in]{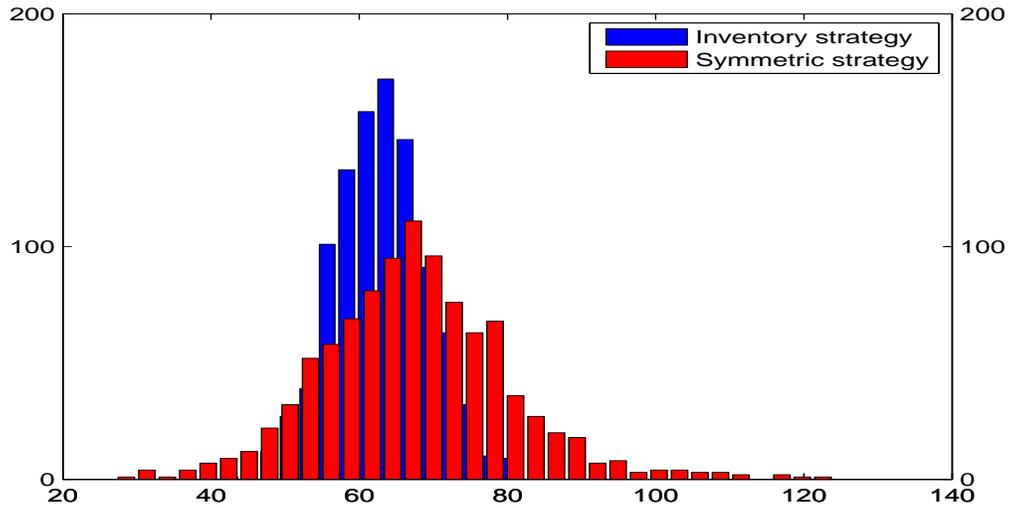}
\caption{A Comparison of two different strategies under the setting with market impact.}
\end{figure}

\section*{Acknowledgements}

This research work was supported by Research Grants Council of Hong Kong under Grant Number 17301214 and HKU CERG Grants and Hung Hing Ying Physical Research Grant.

\end{spacing}
\end{document}